\documentclass{vldb}
\usepackage{mathbbold}
\usepackage{multirow}
\usepackage{graphicx}
\usepackage{balance}  % for  \balance command ON LAST PAGE  (only there!)
\usepackage{amssymb}
\usepackage{amsmath}
\usepackage{latexsym}
\usepackage{graphicx}
\usepackage{subfigure}
\usepackage{psfrag}
\usepackage{PageSetting}
\usepackage{wrapfig,epsfig}
\usepackage{algorithm}
\usepackage{algorithmic}
\newtheorem{remark}{Remark}
\newcommand{\btab}{\begin{tabbing}}
\newcommand{\etab}{\end{tabbing}}

\begin{document}

% ****************** TITLE ****************************************

\title{Randomization Resilient To Sensitive Reconstruction}

% possible, but not really needed or used for PVLDB:
%\subtitle{[Extended Abstract]
%\titlenote{A full version of this paper is available as\textit{Author's Guide to Preparing ACM SIG Proceedings Using \LaTeX$2_\epsilon$\ and BibTeX} at \texttt{www.acm.org/eaddress.htm}}}

% ****************** AUTHORS **************************************


%\numberofauthors{8} %  in this sample file, there are a *total*
%% of EIGHT authors. SIX appear on the 'first-page' (for formatting

\numberofauthors{2}
\author{
% 1st. author
\alignauthor Ke Wang \ \ \ \ \ \ \ \ \ \ \ \  Chao Han \\
       \affaddr{School of Computing Science}\\
       \affaddr{Simon Fraser University}\\
       \affaddr{British Columbia. Canada}\\
       \email{{wangk,hanchao}@cs.sfu.ca}
% 2nd. author
%\alignauthor Chao Han\\
%       \affaddr{School of Computing Science}\\
%       \affaddr{Simon Fraser University}\\
%       \affaddr{British Columbia. Canada}\\
%       \email{hanchaoh@cs.sfu.ca}
% 3rd. author
\alignauthor Ada Waichee Fu\\
       \affaddr{Department of Computer Science and Engineering}\\
       \affaddr{Chinese University of Hong Kong}\\
     %  \affaddr{Hong Kong, China}\\
       \email{adafu@cse.cuhk.edu.hk}
       }
%\and  % use '\and' if you need 'another row' of author names
%% 4th. author
%\alignauthor Lawrence P. Leipuner\\
%       \affaddr{Brookhaven Laboratories}\\
%       \affaddr{Brookhaven National Lab}\\
%       \affaddr{P.O. Box 5000}\\
%       \email{lleipuner@researchlabs.org}
%% 5th. author
%\alignauthor Sean Fogarty\\
%       \affaddr{NASA Ames Research Center}\\
%       \affaddr{Moffett Field}\\
%       \affaddr{California 94035}\\
%       \email{fogartys@amesres.org}
%% 6th. author
%\alignauthor Charles Palmer\\
%       \affaddr{Palmer Research Laboratories}\\
%       \affaddr{8600 Datapoint Drive}\\
%       \affaddr{San Antonio, Texas 78229}\\
%       \email{cpalmer@prl.com}
%}
% There's nothing stopping you putting the seventh, eighth, etc.
% author on the opening page (as the 'third row') but we ask,
% for aesthetic reasons that you place these 'additional authors'
% in the \additional authors block, viz.
%\additionalauthors{Additional authors: John Smith (The Th{\o}rv\"{a}ld Group, {\texttt{jsmith@affiliation.org}}), Julius P.~Kumquat
%(The \raggedright{Kumquat} Consortium, {\small \texttt{jpkumquat@consortium.net}}), and Ahmet Sacan (Drexel University, {\small \texttt{ahmetdevel@gmail.com}})}
% Just remember to make sure that the TOTAL number of authors
% is the number that will appear on the first page PLUS the
% number that will appear in the \additionalauthors section.
\maketitle

\begin{abstract}
With the randomization approach, sensitive data items of records are randomized to protect privacy of individuals while allowing the distribution information to be reconstructed for data analysis.
In this paper, we distinguish between reconstruction that has potential privacy risk, called \emph{micro reconstruction}, and reconstruction that does not, called \emph{aggregate reconstruction}.
We show that the former could disclose sensitive information about a target individual, whereas the latter is more useful for data analysis than for privacy breaches. To limit the privacy risk of micro reconstruction, we propose a privacy definition, called \emph{$(\varepsilon,\delta)$-reconstruction-privacy}.
Intuitively, this privacy notion requires that micro reconstruction has a large error with a large probability. The promise of this approach is that micro reconstruction is more sensitive to the number of independent trials in the randomization process than aggregate reconstruction is; therefore, reducing the number of independent trials helps achieve $(\varepsilon,\delta)$-reconstruction-privacy while preserving the accuracy of aggregate reconstruction. We present an algorithm based on this idea and evaluate the effectiveness of this approach using real life data sets.
\end{abstract}

\section{Introduction}
Randomization is one of the promising approaches in
privacy-preserving data mining. With this approach,
sensitive data items in records are randomized to protect the privacy of individuals while allowing the distribution information to be reconstructed with reasonable accuracy. An early
use of randomization is \emph{randomized response} (RR) for
collecting responses on sensitive questions \cite{W65}. For example, to find the percentage of employees stealing from the company, the employer asks each employee the question ``do you steal from the company?".
To prevent linking the responder to his/her sensitive response, each employee submits the true answer (``Yes or ``No") with a certain \emph{retention probability} $p$ and submits an answer chosen from $\{Yes, No\}$ at random with probability $(1-p)/2$. This type of randomization, also called \emph{input perturbation}, is extended to categorical values in \emph{privacy preserving data mining} for mining association rules \cite{E03,AST05,E02,R02}. Randomization is also studied in \emph{privacy preserving data publishing} where a data publisher has collected the original data $D$ and wants to release a sanitized version $D^*$ for data mining \cite{AH05,ZD08,RHS07,Xiao09,CW10}.
%Most works either consider a single attribute or assume that all attributes are sensitive.

In this paper, we consider the data publishing scenario in which the data set $D$ contains both non-sensitive attributes (e.g., age, gender, etc.) and a sensitive attribute (e.g., disease), as in most realistic settings. We assume that an adversary has named a \emph{target individual}, $t$, whose record is contained in $D$, and has figured out somehow the
non-sensitive attributes of $t$. The adversary's goal is to infer the sensitive attribute of $t$. To preserve the privacy of individuals, the
sensitive attribute value in each record is randomized following a certain retention
probability $p$, while allowing reconstruction of distribution information
such as the count of records in $D$ satisfying a given predicate $\varphi$. We show that, with the help of non-sensitive attributes, the adversary could reconstruct the distribution of the sensitive attribute for a target individual, even if major privacy definitions are satisfied. If this distribution is skewed, the target individual's privacy is breached. This attack is termed ``reconstruction attack".

\subsection{Reconstruction Attacks}
One major privacy definition is limiting the change in adversary's
confidence in the sensitive value $x$ of a given record as a result of interacting
with or exposure to the database. For example, the
\emph{$\rho_1$-$\rho_2$ privacy} proposed in \cite{E03} states that
if the prior probability $\Pr[X=x]$ is not more than $\rho_1$, the
posterior probability $\Pr[X=x | Y=y]$, given the published data $D^*$, should not be more than
$\rho_2$, where $\rho_1<\rho_2$ and $X$ and $Y$ are the variables
for the original and perturbed sensitive values in a record, respectively. In
the literature \cite{E03,AH05,AST05,Xiao09,CW10}, $\Pr[X=x]$ is
measured by the fraction of records with $X=x$ in the \emph{whole}
table $D$, and $\Pr[X=x | Y=y]$ is measured by the fraction of
records with $X=x$ among the records with $Y=y$ in the \emph{whole} table $D^*$.
Precisely, \[\Pr[X=x\mid Y=y]=\frac{\Pr[X=x]\cdot p[x\rightarrow
y]}{\sum_x \Pr[X=x]\cdot p[x\rightarrow y]} \] where $p[x\rightarrow
y]$ is the probability that $x$ is perturbed to $y$, and
can be determined by the retention probability $p$.
Note that these measurements do not take into account the non-sensitive attributes of records in $D$ or the acquired non-sensitive information about the target individual $t$. The next example shows that with non-sensitive information, the adversary could infer the sensitive information of $t$ with a probability higher than $\rho_2$, even if $\rho_1$-$\rho_2$ privacy is ensured.

\begin{example}[Attacks on $\rho_1$-$\rho_2$ privacy]\label{example1}
Let $D$ contain $10\times k$ records over the sensitive attribute
$Disease$ and the non-sensitive attributes $\{Gender,Age\}$, where $k$ is an integer and $Disease$
has the domain $\{x_1,\cdots,x_{10}\}$. Suppose that $k$ records in $D$ have
$Gender=M$ and $Age=30$, all of which have the value $x_1$ for $Disease$. Let $g$ denote this set of records. $x_2$-$x_{10}$ are uniformly distributed among the remaining $9\times k$ records in $D$. Note, for $1\leq i\leq 10$, $\Pr[X=x_i]=10\%$, and 0.1-0.5 privacy ensures $\Pr[X=x_i | Y=y]\leq 50\%$ for all $x_i$. This level of privacy can be achieved by retaining the original value in a record with probability $50\%$ and perturbing $x_i$ randomly to a different value (i.e., $\{x_2,\cdots,x_{10}\}$) with probability $(1-0.5)/9$ \cite{E03,AH05,CW10}. Let $g^*$ denote the randomized version of $g$.
%Note that $D^*$ can be either collected from each individual as in input perturbation or produced by the data publisher as in data publishing.

Suppose that an adversary wants to infer the disease of the target individual $t=Bob$ having the non-sensitive information $Gender=M$ and $Age=30$. The adversary could estimate the (relative) frequencies of $x_1,\cdots, x_{10}$ in $g$ based on $g^*$, instead of $D^*$, because all other records in $D^*$ do not match $t$'s non-sensitive information. Let $\langle F'_1,\cdots,F'_{10} \rangle$ be the estimated frequencies in $g$. For a sufficiently large $g$ (by using a large $k$) and a reasonable estimator such as the \emph{maximum likelihood estimator (MLE)}, $F'_1$ will be sufficiently close to the true frequency $f_1$ \cite{AST05}, which is 100\%. Consequently, the adversary is able to infer that $t$ has the disease $x_1$ with a probability larger than $\rho_2=0.5$.

%In general, as the size of $g_{M,30}$ increases, $\Pr\left[\left|\frac{f'_1-f_1}{f_1}\right|\leq \varepsilon\right]\geq \delta$ holds for a small $\varepsilon$ and a large $\delta$.
\end{example}

%\begin{table}
%\centering \caption{An Example for Reconstruction Attack}
%\begin{tabular}{|c|c|} \hline
%$NA$ & $SA$ \\ \hline $\langle M,20 \rangle$ & $x_1$\\
%\hline \$ & 4 in 5 & Used in business\\ \hline
%$\Psi^2_1$ & 1 in 40,000& Unexplained usage\\
%\hline\end{tabular}
%\end{table}

A recent breakthrough in privacy definition is \emph{differential privacy}
\cite{Dwork06}. The idea is hiding the presence or
absence of a participant in the database by making two neighbor data sets (nearly) equally probable for giving the produced query answer.
Precisely, the $\lambda$-differential privacy mechanism ensures that, for any two data sets $D$ and $D'$ differing on at most one record, for all queries $Q$, and for all query
outputs $o'$,
\[\Pr[K(D,Q) = o'] \leq exp(\lambda) \Pr[K(D',Q) = o']
\]
With a small $\lambda$, $exp(\lambda)$ is close to 1, so $D$ and
$D'$ are almost equally likely to be the underlying database that
produces the final output of the query. To ensure this property, the
$\lambda$-differential privacy mechanism adds the noise $\xi$ to the true answer $o$ and publishes the noisy answer $o'=o +\xi$, where $\xi$ follows the Laplace distribution
$Lap(b)=\frac{1}{2b}exp(-\frac{|\xi|}{b})$, $b=1/\lambda$.
%Note that the relative error
%$\frac{\xi}{o}$ of the noisy answer gets smaller as the true answer
%$o$ gets larger.
The next example shows that such noisy answers can be exploited
to estimate the likelihood of the sensitive value for a target individual.

\begin{example}[Attacks on differential privacy]\label{example2}
Consider the $D$ and $t$ again in Example
\ref{example1}. An adversary could infer the distribution of Disease for $t$
by issuing two queries $Q_1$ and $Q_2$: $Q_1$ asks for the count of
records that satisfy ``$Gender=M \wedge Age=30$" and gets the noisy answer $o'_1=o_1+\xi_1$, and $Q_2$ asks for the count of records that satisfy ``$Gender= M  \wedge Age=30 \wedge Disease=x_1$" and gets the noisy answer $o'_2=o_2+\xi_2$, where $o_i$ are the true answers
and $\xi_i$ are the noises added, $i=1,2$. Note that the relative error
$\frac{\xi_i}{o_i}$ gets smaller as the true answer
$o_i$ gets larger, because  $\xi_i$ has the zero mean and the variance $2b^2$, where $b=1/\lambda$ is a constant for a given $\lambda$-differential privacy mechanism.
Therefore, as the answer $o_1$ increases, $o'_2/o'_1$ approaches $o_2/o_1$, the
fraction of records having $x_1$ among the records that share the gender and age with $t$. This discloses the disease $x_1$ of $t$ because $o_2/o_1 = 100\%$.
\end{example}

In these examples, randomized data or noisy query answers are used to reconstruct the distribution of sensitive information for a target individual, even though strong privacy definitions are satisfied. If such reconstruction is accurate and if the true distribution is skewed, as in these examples, the reconstructed distribution discloses the sensitive information of the target individual with a high probability. This attack is powerful in that it works on different types of randomization techniques and data sharing scenarios, i.e., the random value replacement in Example \ref{example1}, through either input perturbation or data publishing; random noise addition to query answers in Example \ref{example2}, also known as \emph{output perturbation}.

\subsection{Contributions}
The contributions in this work are as follows.

\begin{itemize}
\item
%Traditionally reconstruction of data distribution from randomized data is considered as utility.
For the first time, we consider the implication of non-sensitive attributes on sensitive reconstruction of data distribution from randomized data. We distinguish two types of reconstruction: The \emph{micro reconstruction} seeks to reconstruct the distribution of the sensitive attribute in a set of records that fully match a target individual on \emph{all} non-sensitive attributes; the \emph{aggregate reconstruction} aims to reconstruct the distribution in a set of records that only partially match a target individual. We argue that micro reconstruction is all we have to be concerned with about privacy risk.

\item
To address the privacy risk of micro reconstruction, we propose a notion of \emph{$(\varepsilon,\delta)$-reconstruction-privacy} to ensure a minimum value on the tail probabilities of micro reconstruction error. We present a \emph{bound conversion theorem} that converts between a bound on tail probabilities of a random variable and a bound on tail probabilities of reconstruction error, which allows us to leverage the Chernoff bound to develop a testable instantiation of $(\varepsilon,\delta)$-reconstruction-privacy. Since the bound conversion theorem does not hinge on the particular form of bounds, our approach can be instantiated to other upper bounds and modified to constrain lower bounds of tail probabilities.

\item
 The promise of this approach is that micro reconstruction is more sensitive to the number of independent trials in the randomization process than aggregate reconstruction, analogous to the fact that the first 10 coin flips are more critical for the estimation of head probability than the second 10 coin flips. We leverage this difference to design an algorithm for achieving $(\varepsilon,\delta)$-reconstruction-privacy while preserving the utility of aggregate reconstruction.

\item Empirical evaluation on real life data sets presents two important findings: Firstly, $(\varepsilon,\delta)$-reconstruction-privacy is violated even when major privacy definitions such as $\rho_1$-$\rho_2$ privacy and differential privacy are satisfied. Secondly, the additional information loss incurred for achieving $(\varepsilon,\delta)$-reconstruction-privacy is small.
\end{itemize}

The rest of the paper is organized as follows. Section 2 reviews related work. Section 3 defines the problem studied in this work. Section 4 presents an efficient instantiation of $(\varepsilon,\delta)$-reconstruction-privacy. Section 5 presents the algorithm to achieve $(\varepsilon,\delta)$-reconstruction-privacy.
%Section 6 examines differential privacy in the context of $(\varepsilon,\delta)$-reconstruction-privacy.
Section 6 presents empirical findings. Finally, we conclude the paper.

\section{Related Work}
Two classes of randomization methods have been extensively studied
in the literature: \emph{random perturbation} and \emph{randomized response}.
Random perturbation is primarily used for quantitative data.
For example, Agrawal and Srikant \cite{AS00} build accurate decision
tree classification models on the perturbed data, and Kargupta et
al. \cite{K03} point out that arbitrary randomization can reveal
significant amount of information under certain conditions.
Randomized response is primarily used for categorical
data. Its basic idea was proposed by Warner \cite{W65}, and based on this technique the problem of
mining association rules from disguised data was studied in \cite{E03,E02,R02}. In this paper,
the term ``perturbation" or ``randomization" refers to the randomization for categorical data.

Techniques for probabilistic perturbation have also been
investigated in the statistics literature. The PRAM method
\cite{G97} considers the use of Markovian perturbation matrices. The
disclosure risk is measured by a notion of \emph{expectation ratios},
defined as the ratio of the expected number of records in the
perturbed file with the observed value equal to the value in the
original file, and the expected number of records in the perturbed
file with the observed value not equal to the value in original
file.

Formal definitions of privacy breaches were proposed in
\cite{E03,D03} following the same paradigm: for every record in the
database, the adversary's confidence in the values of the given
record should not significantly increase as a result of interacting
with or exposure to the database. Recent works based on such definitions include
\cite{RHS07,ZD08,AH05,AST05,Xiao09,CW10}. These approaches either consider one attribute (i.e., the sensitive attribute) \cite{ZD08}, or assume that all attributes are sensitive \cite{RHS07,AH05,AST05}, or ignore the role of non-sensitive attributes in the reconstruction of the sensitive attribute in the context of privacy risk \cite{Xiao09,CW10}. Reconstruction of data distribution is traditionally considered as utility. To our knowledge, our work is the first to study such reconstruction as privacy breaches.

An alternative to the randomization approach is the partition based approach in which
the records are partitioned to ensure some sort of balanced distribution of sensitive
data items in each partition \cite{MKG+006,XT06b}.
%Two
%approaches for generating partitions are
%generalization \cite{MKG+006,N07} and bucketization \cite{XT06b}.
The randomization approach, due to its non-deterministic nature, is more robust to
auxiliary information \cite{Kifer09,R10,TaoCorr08}.

The differential privacy mechanism \cite{Dwork06} hides the presence
of a single record in the database by adding random noises to a query answer. As we will see in Section 6,
such noise addition is not sufficient to prevent the adversary from reconstructing the
distribution of sensitive data for a target individual.

\section{Problem Statement}
%This section defines the problems studied in this paper.
%\subsection{Preliminaries}
We assume that the data publisher has collected a table $D(NA,SA)$ on \emph{non-sensitive attributes} $NA=\{A_1,\cdots,A_d\}$ and one \emph{sensitive attribute} $SA$. Each record in the table corresponds to a
participant or individual. For a record $r$ in $D$, $r[NA]$ and $r[SA]$
denote the values of $r$ on $NA$ and $SA$. $|\cdot|$ denotes the
cardinality of a set. The sensitive attribute $SA$ has a discrete domain $\{x_1,\cdots,x_m\}$.
The \emph{count} of $x_i$ refers to the number of records having $x_i$, and the \emph{frequency} of $x_i$ refers to the percentage of records having $x_i$. As in \cite{AH05,E03,AST05,ZD08}, we assume that the $SA$ value in a record is chosen independently at random according to some fixed probability distribution. The publisher allows the
researcher to learn this distribution, but wants to hide the $SA$ value of an
individual record.
%$o_i$ denotes the number of records for $x_i$ in
%$D$ and $f_i$ denotes the frequency $o_i/|D|$, where

%Table \ref{tab1} lists some
%of the notations used in this paper.

%\begin{table}
%\begin{tabular}{| c | c |} \hline
%$D$, $|D|$ & \parbox[t]{6cm}{the raw data and its cardinality} \\
%\hline
% $m$ & \parbox[t]{6cm}{domain size of $SA$} \\ \hline
% $x_i$ & \parbox[t]{6cm}{a sensitive value} \\ \hline
% $o_i$ & \parbox[t]{6cm}{number of occurrence of $x_i$ in $D$} \\ \hline
% $f_i$ & \parbox[t]{6cm}{$o_i/|D|$} \\ \hline
% $f'_i$ & \parbox[t]{6cm}{estimate of $f_i$} \\ \hline
% \end{tabular}
%\caption{Notations} \label{tab1}
%\end{table}

\subsection{Perturbation}
We consider the data publishing scenario where the data publisher wants to publish $D$ for data analysis, but wants to hide the $SA$ value in a record. In the \emph{uniform perturbation} \cite{AH05,E03,AST05,ZD08}, the $SA$ value $x$ in a record is processed by flipping a coin with head probability $0< p < 1$, called \emph{retention probability}.
If the coin lands on heads, $x$ is retained; otherwise, $x$ is
replaced with a random value from the domain of $SA$, where each value is selected with probability $(1-p)/m$. This perturbation process is parameterized by the perturbation matrix
$\mathbb{P}_{m\times m}$:
\begin{equation}\label{E1}
\mathbb{P}_{ji} = \left\{
\begin{array}{l}
p + \frac{1 - p}{m}\;\;\text{if } j$$=$$i\;\text{(retain $x_i$)}\\
\frac{1 - p}{m}\;\;\;\;\;\;\;\;\text{if } j$$\neq$$i\;\text{(perturb $x_i$
to $x_j$)}
\end{array} \right.
\end{equation}
$p +
\frac{1-p}{m}$ is the sum of the probability that $x_i$ is retained and
the probability that $x_i$ is replaced with the same $x_i$. Let $D^*$ contain all
perturbed records. For any subset $S$ of $D$, $S^*$ denotes the same set of records as $S$ in $D^*$. Note $|S^*| = |S|$. The choice of $p$ dictates the trade-off between the privacy concern
of hiding the sensitive value in a record and the utility for reconstructing the distribution of $SA$. The work in \cite{E03,CW10} determines the maximum retention probability $p$ for ensuring a given $\rho_1$-$\rho_2$ privacy \cite{E03} based on $\rho_1,\rho_2$, and $m$.

%where $\rho_1$-$\rho_2$ privacy ensures that if the prior belief about a $SA$ value is no more than $\rho_1$, the posterior belief should not be more than $\rho_1$.

%\begin{remark}\label{remark1}
The above perturbation process has some interesting properties. First, it modifies only the $SA$ attribute, not $NA$ attributes. Therefore, data analysis involving only $NA$ attributes incurs no information loss by accessing the randomized data $D^*$. This is an advantage compared to the differential privacy mechanism \cite{Dwork06} where a query answer will be distorted even if it only involves non-sensitive attributes. Second, the perturbation of a record depends on the original $SA$ attribute in the record, but not on any other records in $D$. Therefore, for any subset $S$ of records from $D$, we can assume that $S^*$ is produced by the same perturbation matrix $\mathbb{P}$. This record independence also implies that insertion and deletion of records on $D$ can be done through insertion and deletion of randomized records on $D^*$.
%\end{remark}

We consider data analysis through answering \emph{count queries}. A count query
 has a predicate $\varphi$ of the form $\wedge (A=a)$, where $A$ is either $SA$ or an attribute in $NA$, and $a$ is a value from the domain of $A$.  The answer to the query is the count of the records in $D$ satisfying $\varphi$. This answer must be estimated using $D^*$. If $\varphi$ contains no equality for $SA$,
the answer on $D^*$ is exactly same as the answer on $D$. If $\varphi$ contains an equality $SA=x_i$,
a \emph{reconstruction} process will be applied to the subset of records in $D^*$ that satisfy $\varphi^{-}$, where $\varphi^{-}$ is $\varphi$ with the equality $SA=x_i$ removed. Let $S^*$ be this subset and let $S$ be the set of corresponding records in $D$. The reconstruction seeks the most likely estimator of the distribution of $SA$ in $S$, denoted by $\overleftarrow{F'}$, given $S^*$ and the perturbation operator $\mathbb{P}$. The answer for the query is estimated by $|S|F'_i$, where $F'_i$ is the component of  $\overleftarrow{F'}$ for $x_i$. The detailed reconstruction will be discussed in Section 4.1.

\subsection{Adversaries and Micro Reconstruction}
We assume that an adversary has named some \emph{target individual}, denoted $t$, whose record is contained in $D$, and has figured out $t$'s values on all non-sensitive attributes $NA$.
%
%To hide the $SA$ value of a record, the $SA$ value in a record is randomized to provide some uncertainty about the true $SA$ value such as $\rho_1$-$\rho_2$ privacy \cite{E03}.
To infer the $SA$ value of $t$, the adversary needs to reconstruct the frequencies $\overleftarrow{F'}$ of $SA$ values from the randomized data $D^*$. Given the knowledge about $t$'s information on all non-sensitive attributes in $NA$, the adversary would focus the reconstruction process on the records in $D^*$ that match all $t$'s non-sensitive attributes. The next definition formalizes this reconstruction.

\begin{definition}[Micro/Aggregate Reconstruction]\label{micro}
A \emph{micro group} is a set of the records in $D$ that agree on all attributes in $NA$. The \emph{micro reconstruction} seeks to reconstruct the distribution of $SA$ in a micro group. For a target individual $t$, $g_t$ denotes the micro group containing $t$'s record and $g_t^*$ denotes the set of corresponding records in $D^*$. An \emph{aggregate group} is a set of the records in $D$ that agree on zero or more but not all attributes in $NA$. The \emph{aggregate reconstruction} seeks to reconstruct the distribution of $SA$ in an aggregate group.
\end{definition}

%For each target individual $t$, $g_t$ denotes the micro group that agrees with $t$'s non-sensitive attributes $NA$ and $g_t^*$ denotes the set of corresponding records in $D^*$.

The intent of distinguishing these two types of reconstruction
is that micro reconstruction is all we have to be concerned with about privacy risk - aggregate reconstruction
does not present privacy risk. The next example illustrates this point.

\begin{example}\label{example3}
Let $NA=\{Gender,Job\}$ and $SA=Disease$. Consider a target individual $t$ (say Bob) with $Gender=Male$ and $Job=Teacher$. The micro group for $t$, $g_{t}$, contains all records in $D$ with $Gender=Male$ and $Job=Teacher$. The micro reconstruction for $t$ seeks to reconstruct the distribution of $SA$ in $g_t$ using the published $g_t^*$. This reconstruction is most relevant to $t$ because $g_t$ contains all and only the records in $D$ that match $t$'s non-sensitive information. In contrast, aggregate reconstruction involves records that do not match $t'$ information in at least one of $Gender$ and $Job$, such as (1) all records for $Job=Teacher$, or (2) all records for $Gender=Female$, or (3) all records for $Gender=Female \wedge Job=Teacher$, or (4) all records in $D$. These reconstructions are less relevant to $t$ because they are based on more records that do not belong to $t$. For example,
a high estimated frequency of Breast Cancer in (1) does not mean that $t$ has a high chance of getting Breast Cancer because most occurrences of Breast Cancer actually come from female teachers.
\end{example}

In the above, we distinguish two types of reconstruction based on the set of records in which the data distribution is estimated. For each type of reconstruction, we can distinguish two types of estimates based on the records used to derive the estimate.
In Example \ref{example3}, we estimate the distribution of $SA$ in $g_t$ based on the records in $g_t^*$. Alternatively, we can treat $g_t$ as the difference $X-Y$ of two sets $X$ and $Y$, where $S \subseteq X$ and $Y=X-S$, and estimate the distribution of $SA$ in $g_t$ based on the estimates for $X$ and $Y$. For example, for the set of male teachers, $g_t$, $g_t =X-Y$, where $X$ is the set of all teacher records in $D$ and $Y$ is the set of all female teacher records in $D$. If $F'_X$ and $F'_Y$ are the estimated frequencies of Breast Cancer in $X$ and $Y$ based on $X^*$ and $Y^*$, respectively, we can estimate the frequency of Breast Cancer in $g_t$ by $(F'_X |X|-F'_Y |Y|)/|g_t|$. The next definition summarizes these two types of estimation.

\begin{definition}[Local/Global Estimates]\label{estimate}
For any subset $S$ of $D$ and any $SA$ value $x$, the \emph{local estimate} for $x$ wrt $S$ is based on the information in $S^*$, and a \emph{global estimate} for $x$ wrt $S$ is given by $(F'_X |X|-F'_Y |Y|)/|S|$, where $S \subseteq X$, $Y=X-S$, and $F'_X$  and $F'_Y$ are the local estimates of $x$ wrt $X$ and $Y$, respectively.
\end{definition}

Every local estimate is a global estimate in the special case of $X=S$ and $Y=\emptyset$. At first glance, there is a temptation for considering global estimates because the use of a superset $X^*$ is in favor of accurate reconstruction. However, we will show that all global estimates are in fact equal to the local estimate in Section 4.1.

\begin{table}[h]
\centering
\caption{Notations}
\begin{tabular}{|l|l|} \hline
\textbf{Symbols} & \textbf{Meaning} \\ \hline
$m$ & the domain size $|SA|$ \\ \hline
$t$ & a target individual \\ \hline
$S$ & a subset of records in $D$ \\ \hline
$S^*$ & the corresponding set of $S$ in $D^*$ \\ \hline
$g$ & a micro group \\ \hline
$g^*$ & the corresponding set of $g$ in $D^*$ \\ \hline
$x$ & a domain value of $SA$ \\ \hline
$f$ & the frequency of $x$ in $S$ \\ \hline
$O^*$ & the variable for the observed count of $x$ in $S^*$ \\ \hline
$F'$ & the variable for the local estimate of $f$ \\ \hline
$\overleftarrow{f}, \overleftarrow{F'}$, $\overleftarrow{O^*}$ & the column-vectors of $f$, $F'$, $O^*$ \\ \hline
$\mathbb{P}$ & the perturbation matrix in Equation (1) \\ \hline
$p$ &  the retention probability \\ \hline
\end{tabular}
\label{table:notation}
\end{table}

\subsection{Problems}

%
%\textbf{Notation}. $m$ denotes the domain size $|SA|$.
%$S$ denotes any subset of records in $D$ and $S^*$ denotes the corresponding set of records in $D^*$. For each $SA$ value $x$,
%$f$ denotes the frequency in $S$, $O^*$ denotes the observed count in $S^*$, and $F'$ denotes the local estimate of $f$. Note that $O^*$ and $F'$ are variables over the random instances of $S^*$.
%$\overleftarrow{f}, \overleftarrow{F'}$, and $\overleftarrow{O^*}$ are the column-vectors of $f$, $F'$, and $O^*$ for all $SA$ values. $\mathbb{P}$ and $p$ denote the perturbation matrix and retention probability defined in Equation (1). $\mathbb{P}^{-1}$ is the inversion of $\mathbb{P}$.

We are now ready to define the problem we will study. We adapt the notation in Table \ref{table:notation} in the rest of the paper. For each target individual $t$, the micro reconstruction for $t$
reconstruct the distribution of $SA$ most relevant to $t$. If the distribution is skewed and if the reconstruction is accurate, $t$'s $SA$ information will be disclosed. To limit this privacy risk,
the next definition formalizes a privacy definition through bounding the accuracy of micro reconstruction.

\begin{definition}[$(\varepsilon, \delta)$-reconstruction-privacy]\label{privacy}
For a micro group $g$, $g^*$ is $(\varepsilon, \delta)$-reconstruction-private, where $\varepsilon\geq 0$ and $\delta \in [0,1]$, if for each $SA$ value $x$ occurring in $g$, whenever $\Pr\left[\frac{F'-f}{f}>
\varepsilon\right]< U$ or $\Pr\left[\frac{F'-f}{f}< -\varepsilon\right]< L$,
$\delta\leq min\{U,L\}$, where $f$ is the frequency of $x$ in $g$
and $F'$ is the variable for a global estimate of $f$ over the random instances of $g^*$.
$D^*$ is $(\varepsilon,
\delta)$-reconstruction-private if $g^*$ is $(\varepsilon, \delta)$-reconstruction-private for every micro group $g$.
\end{definition}

%\begin{corollary}\label{privacy-definition}
%If $D^*$ is $(\varepsilon, \delta)$-reconstruction-private, for any target individual $t$ and for any $SA$ value $x$, the adversary cannot prove either $\Pr\left[\frac{F'-f}{f}\geq \varepsilon\right]<
%\delta$ or $\Pr\left[\frac{F'-f}{f}\leq -\varepsilon\right]<\delta$, where $f$ is
%the frequency of $x$ in the micro group $g_t$ and $F'$ is the variable for a global estimate of $f$ over all states of $g_t^*$
%\end{corollary}

\begin{remark}\label{remark1}
$(\varepsilon, \delta)$-reconstruction-privacy ensures that the (best) upper bounds on tail probabilities for micro reconstruction error greater than $\varepsilon$ or smaller than $-\varepsilon$ are not smaller than $\delta$. In this sense, the adversary has difficulty to lower the probabilities of a large estimation error. The larger the parameters  $\varepsilon$ and $\delta$ are, the greater this difficulty is and the more secure the published data is. In this definition, $\delta$ is a constraint on the upper bounds of tail probabilities (i.e., $U$ and $L$). This formulation allows us to leverage the extensive research on upper bounds of tail probabilities in the literature. Alternatively, $\delta$ could be a constraint on the lower bounds of tail probabilities if such bounds are available, and from Theorem \ref{UL-bound}, our approach does not hinge on whether $U$ and $L$ are upper bounds or lower bounds. In this definition, we consider the estimate $F'$ estimated from randomized data. In Section 6.3, we will show that the same privacy notion can be applied to $F'$ estimated from noisy query answers such as those produced by the differential privacy mechanism.
\end{remark}

%$(\varepsilon, \delta)$-reconstruction-privacy is a constraint on the error for micro reconstruction, but it places no constraint on aggregate reconstruction. We shall leverage this difference to achieve $(\varepsilon, \delta)$-reconstruction-privacy while preserving the accuracy of aggregate reconstruction.

%Randomization, specified by a retention probability $p$, provides uncertainty about the $SA$ value in a record (such as $\rho_1$-$\rho_2$ privacy). $(\varepsilon,\delta)$-reconstruction-privacy further ensures that the distribution of $SA$ for any target individual cannot be reconstructed from the randomized data $D^*$.

\begin{definition}[The Problem]
Given a data set $D$, a retention probability $p$ for randomization, $\varepsilon$, and $\delta$, where $\varepsilon\geq 0$ and $\delta \in [0,1]$,
we want to produce a randomized version $D^*$ that satisfies $(\varepsilon, \delta)$-reconstruction-privacy while information for aggregate reconstruction is preserved.
\end{definition}

Two main problems are to be solved: how to test if $(\varepsilon, \delta)$-reconstruction-privacy is satisfied, and how to achieve $(\varepsilon, \delta)$-reconstruction-privacy on a given data set. We answer the first question in Section 4 and answer the second question in Section 5.

\section{Testing Privacy}
We first present an estimation technique for $F'$ and then present a probabilistic bound for the estimation error of $F'$. In the discussion below, the reader is referred to the notations in Table \ref{table:notation}.

\subsection{Maximum Likelihood Estimator}
We adapt the \emph{maximum likelihood estimator (MLE)} as our model of local estimates.
The next theorem follows from Theorem 2 in \cite{AST05}.

\begin{theorem}[Theorem 2, \cite{AST05}]\label{MLE}
For a subset of records $S$ and any $SA$ value $x$, $\overleftarrow{F'}$ computed by $\mathbb{P}^{-1}\cdot \frac{\overleftarrow{O^*}}{|S|}$ is the maximum likelihood estimator (MLE) of $\overleftarrow{f}$ in $S$, under the constraint $\Sigma F' = 1$, where $\Sigma$ is over all elements of $\overleftarrow{F'}$.
\end{theorem}

In the rest of the paper, $\overleftarrow{F'}$ denotes the MLE computed by Theorem \ref{MLE}. The presence of  the matrix inversion $\mathbb{P}^{-1}$ makes it troublesome to compute $\overleftarrow{F'}$ and develop a probabilistic error bound for $\overleftarrow{F'}$. The next lemma gives an efficient computation of $\overleftarrow{F'}$.

\begin{lemma}[Computing $\overleftarrow{F'}$] \label{f'}
For any subset $S$ of $D$ and any $SA$ value $x$, (i) $E[O^*]=|S|(fp+(1-p)/m)$, (ii) $F'=\frac{O^*/|S| - (1-p)/m}{p}$, and
(iii) $E[F']=f$.
\end{lemma}
\begin{proof}
(i) Let $X_k$ be independent and identically distributed (i.i.d.)
indicator variables for the event that the $k$-th row in $S^*$ has the $SA$ value
$x$. $O^*=\Sigma_k X_k$. From the matrix $\mathbb{P}$ in Equation (\ref{E1}), if the $k$-th row in $S$ has $x$,
$X_k=1$ with probability $p+(1-p)/m$, and if the
$k$-th row in $S$ does not have $x$, $X_k=1$ with probability $(1-p)/m$. So $E[O^*]=|S|f(p+(1-p)/m)+|S^*|(1-f)(1-p)/m)=|S|(fp+(1-p)/m)$.
This shows (i).

(ii) From Theorem \ref{MLE}, $\overleftarrow{F'}=\mathbb{P}^{-1}\cdot \frac{\overleftarrow{O^*}}{|S|}$. Let $[\alpha]_m$ denote a column-vector of the constant $\alpha$ of the length $m$. We have
\[ \frac{\overleftarrow{O^*}}{|S|} =  \mathbb{P}\cdot \overleftarrow{F'} = p\overleftarrow{F'} + [\frac{1-p}{m} \Sigma F']_m  = p\overleftarrow{F'} + [\frac{1-p}{m}]_m  \]
The last equation holds because $\sum F'=1$ (Theorem \ref{MLE}). Thus, $\frac{O^*}{|S|} =  p F' + \frac{1-p}{m}$, equivalently, $F'=\frac{O^*/|S| - (1-p)/m}{p}$, as required for (ii).

(iii) Taking the mean on both sides of $F'=\frac{O^*/|S| - (1-p)/m}{p}$, we get
$E[F']=\frac{E[O^*]/|S| - (1-p)/m}{p}$. Substituting $E[O^*]$ in (i) into the last equation and simplifying, we get $E[F']=f$. This shows
(iii).
\end{proof}

From Lemma \ref{f'}(ii), $F'$ can be computed directly from the observed count $O^*$ without computing the matrix inversion $\mathbb{P}^{-1}$. From Lemma \ref{f'}(iii), $F'$ is an unbiased estimator of $f$.
The next lemma shows that, for the MLE model of local estimates, all global estimates are equal to the local estimate.

\begin{lemma}\label{global}
For any subset $S$ of $D$ and any $SA$ value $x$, every global estimate for $x$ wrt $S$ is equal to the MLE for $x$ wrt $S$.
\end{lemma}
\begin{proof}
From Definition \ref{estimate}, every global estimate wrt $S$ has the form $\frac{F'_X|X|-F'_Y|Y|}{|S|}$, where $S \subseteq X \subseteq D$ and
$Y =X - S$, and $F',F'_X,F'_Y$ are the MLEs wrt
$S,X,Y$, respectively. Let $S^*,X^*,Y^*$ be the sets of records in $D^*$ corresponding to $S,X,Y$, and let $O^*,O^*_X,O^*_Y$ be the variables for the counts of $x$ in $S^*,X^*,Y^*$, respectively.
From Lemma \ref{f'}(ii), $F'_X=\frac{O^*_X/|X| -(1-p)/m}{p}$ and $F'_Y=\frac{O^*_Y/|Y| -(1-p)/m}{p}$.
Substituting these into $\frac{F'_X|X|-F'_Y|Y|}{|S|}$, noting $|S|=|X|-|Y|$ and $O^*=O^*_X-O^*_Y$, we get $\frac{O^*/|S|-(1-p)/m}{p}$, which is equal to the MLE $F'$ given by Lemma \ref{f'}(ii). This shows that every global estimate for $x$ is equal to the MLE $F'$ for $x$. \end{proof}

Consequently, it suffices to consider only local estimates. The next definition refines Definition \ref{privacy} by considering only local estimates and will be used in the remaining discussion about $(\varepsilon, \delta)$-reconstruction-privacy.

\begin{definition}[$(\varepsilon, \delta)$-reconstruction-privacy (Refined)]\label{refined-privacy}
For any micro group $g$, $g^*$ is $(\varepsilon, \delta)$-reconstruction-private, where $\varepsilon\geq 0$ and $\delta \in [0,1]$, if for each $SA$ value $x$ occurring in $g$, whenever $\Pr\left[\frac{F'-f}{f}>
\varepsilon\right]< U$ or $\Pr\left[\frac{F'-f}{f}< -\varepsilon\right]< L$, then
$\delta\leq min\{U,L\}$, where $f$ is the frequency of $x$ in $g$
and $F'$ is the variable for the MLE of $f$ under the constraint $\Sigma F'=1$.
\end{definition}

\subsection{Probabilistic Error Bounds}
A remaining question is how to bound $\Pr\left[\frac{F'-f}{f}> \varepsilon\right]$ and $\Pr\left[\frac{F'-f}{f}< -\varepsilon\right]$. We leverage tail probabilities of random variables in the literature to develop such bounds. Recall that $O^*$ is the observed count of a $SA$ value and $F'$ is the reconstructed frequency of a $SA$ value. The next theorem gives a conversion between a probabilistic bound for $F'$ and a probabilistic bound for $O^*$.

\begin{theorem}[Bound Conversion]\label{UL-bound}
Consider any subset $S$ of $D$ and any $SA$ value $x$. Let $\mu=E[O^*]$. For any upper tail bound function $U(\theta,\mu)$ and lower tail bound function $L(\theta,\mu)$, and for any comparison operator $\bigoplus$ (i.e., $<$ or $>$),
\begin{enumerate}
\item $\Pr\left[\frac{O^*-\mu}{\mu}> \theta\right] \bigoplus U(\theta,\mu)$ if and only if $\Pr\left[\frac{F'-f}{f}> \varepsilon\right]$ $\bigoplus  U(\frac{\varepsilon |S|pf}{\mu},\mu)$;
\item $\Pr\left[\frac{O^*-\mu}{\mu}< -\theta\right] \bigoplus L(\theta,\mu)$ if and only if $\Pr\left[\frac{F'-f}{f} <
-\varepsilon\right]$ $\bigoplus L(\frac{\varepsilon |S|pf}{\mu},\mu)$.
\end{enumerate}
\end{theorem}
\begin{proof}
We show (1) only because the proof for (2) is similar. From Lemma \ref{f'}(ii), $F'=\frac{O^*/|S| - (1-p)/m}{p}$, $O^*=|S|(F'p+(1-p)/m)$, and from Lemma \ref{f'}(i),
$\mu=|S|(fp+(1-p)/m)$. So
\begin{eqnarray*}
 \frac{O^*-\mu}{\mu} > \theta & \Leftrightarrow & O^*-\mu > \theta\mu\\
 & \Leftrightarrow & |S|p(F'-f )> \theta\mu\\
 & \Leftrightarrow & \frac{F'-f}{f}> \frac{\theta \mu}{|S|pf} =\varepsilon.
 %& \Leftrightarrow & \frac{F'-f}{f}\geq \varepsilon.
\end{eqnarray*}
%The last step follows by letting $\varepsilon=\frac{\theta \mu}{|S|pf}$.
These rewriting implies that the probabilities on the two sides of (1) are equal.
Then (1) follows because $\theta=\frac{\varepsilon|S|pf}{\mu}$.
\end{proof}

%\begin{remark}\label{remark2}
From Theorem \ref{UL-bound}, if we have a tail probability bound for the error of $O^*$ (i.e., $U(\theta,\mu)$ and $L(\theta,\mu)$), we immediately have a tail probability bound for the error of $F'$ (i.e., $U(\frac{\varepsilon |S|pf}{\mu},\mu)$ and $L(\frac{\varepsilon |S|pf}{\mu},\mu)$). Moreover, if the bound for $O^*$ is the best, the corresponding bound for $F'$ is also the best (otherwise, a better bound for $O^*$ can be obtained from Theorem \ref{UL-bound}). Importantly, the bound conversion does not hinge on the particular form of the bound functions $U$ and $L$. This generality allows us to adapt to the best bounds $U$ and $L$ available for $O^*$ to get the best bounds for $F'$.

%Since  $O^*$ is the sum of independent indicator variables $X_k$ for the event that the $k$-th row in $S^*$ has the $x$ value, we can leverage the rich literature on tail probabilities for such variables.

%\end{remark}

%The Markov's inequality has the form $\Pr\left[X\geq t\right]\leq \frac{\mu}{t}$, where $\mu = E[X]$ and $X$ is any non-negative random variable $X$. This inequality applies to $O^*$.
%The Chebyshev's inequality uses knowledge of the standard deviation $\sigma$ to
%give a tighter bound: for any real
%number $k > 0$,
%\[
%    \Pr\left[\frac{|X-\mu|}{\mu}\geq \frac{k\sigma}{\mu}\right] \leq \frac{1}{k^2}.
%\]
%Letting $\theta=\frac{k\sigma}{\mu}$, we get
%\[
%    \Pr\left[\frac{|X-\mu|}{\mu}\geq \theta\right] \leq (\frac{\sigma}{\theta \mu})^2.
%\]

There is a rich literature on the upper bounds for tail probabilities of random variables.
The Markov's inequality applies to any non-negative random variable, therefore, applies to $O^*$. The Chebyshev's inequality uses knowledge of the standard deviation to give a tighter bound. However, these bounds are very poor for random variables that fall off exponentially with distance from the mean. The Chernoff bound, due to \cite{Chernoff}, gives exponential
fall-off of probability with distance from the mean. The critical
condition that is needed for the Chernoff bound is that the random
variable be a sum of independent Poisson trials.

\begin{theorem}[Chernoff Bounds, \cite{Chernoff,MR95}] \label{Chernoff}
Let $X_1,\cdots,X_n$ be independent Poisson trials such that for $1\leq i\leq n$, $X_i \in
\{0,1\}$, $\Pr[X_i=1]=p_i$, where $0<p_i< 1$. Let $X=X_1+\cdots +X_n$ and $\mu =
E[X] = E[X_1] +\cdots + E[X_n]$.
For $\theta \in (0,\infty)$,
\begin{eqnarray}
\Pr\left[\frac{X-\mu}{\mu}> \theta\right]< U _1(\theta,\mu) = \left(\frac{e^{\theta}}{(1+\theta)^{(1+\theta)}}\right)^{\mu}
\end{eqnarray}
and for $\theta \in (0,1]$,
\begin{eqnarray}
\Pr\left[\frac{X-\mu}{\mu}< -\theta\right] < L_1(\theta,\mu) = \left(\frac{e^{-\theta}}{(1-\theta)^{(1-\theta)}}\right)^{\mu}
\end{eqnarray}
\end{theorem}

These full Chernoff bounds are quite tight but can be clumsy to compute. Using the Taylor series expansion $\ln (1 + \theta) = \sum_{i\geq 1} (-1)^{i+1} \frac{\theta^i}{i}$ and ignoring higher order terms, the above bounds can be simplified to the following weaker bounds, which covers 95\% of cases pretty well: For $\theta \in (0,\infty)$,
\begin{eqnarray}
\Pr\left[\frac{X-\mu}{\mu}> \theta\right] < U_2(\theta,\mu) = exp(-\frac{\theta^2}{2+\theta} \mu)
\end{eqnarray}
and for $\theta \in (0,1]$,
\begin{eqnarray}
\Pr\left[\frac{X-\mu}{\mu} < -\theta\right] < L_2(\theta,\mu) = exp(-\frac{\theta^2}{2} \mu).
\end{eqnarray}

The Chernoff bound applies to our variable $O^*$ because $O^*$ is the sum $X_1+\cdots +X_n$, where each $X_i$ is the indicator variable whether the $i$-th row in $S^*$ has a particular $SA$ value $x$, and $E[O^*]=|S|(fp+(1-p)/m)$ (Lemma \ref{f'}). Instantiating the upper bounds $U_i$ and $L_i$ for $O^*$ in Equations (2)-(5) into Theorem \ref{UL-bound}, the next corollary gives the corresponding upper bounds for $F'$.

\begin{corollary}[Upper bounds for $F'$] \label{delta}
Let $U_i$ and $L_i$ be defined in Equations (2)-(5). For $\theta \in (0,\infty)$,
\begin{eqnarray}
\Pr\left[\frac{F'-f}{f} >
\varepsilon\right] < U_i(\theta,\mu)
\end{eqnarray}
and for $\theta \in (0,1]$,
\begin{eqnarray}
\Pr\left[\frac{F'-f}{f}  < -\varepsilon\right] < L_i(\theta,\mu)
\end{eqnarray}
where $\theta=\frac{\varepsilon |S|pf}{\mu}$ and $\mu=|S|(fp+(1-p)/m)$.
\end{corollary}

Corollary \ref{delta} gives the concrete upper bounds $U_i$ and $L_i$ on the tail probabilities of $F'$ based on the Chernoff bound. Since these bounds are public, $(\varepsilon,\delta)$-reconstruction-privacy implies $\delta \leq min\{U_i,L_i\}$. The question is whether $\delta \leq min\{U_i,L_i\}$ is sufficient for $(\varepsilon,\delta)$-reconstruction-privacy, in other words, whether there are tighter (i.e., smaller) upper bounds than $U_i$ and $L_i$. To answer this question, we observe from Theorem \ref{UL-bound} that any tighter bound for $F'$ would lead to a tighter bound than the Chernoff bound for $O^*$. The fact that the Chernoff bound has been used as the state-of-the-art technique in the past 60 years suggests that it is nontrivial to improve the Chernoff bound. For this reason, we assume that Corollary \ref{delta} gives the best upper bounds for $F'$; however, if better bounds on random variables become available, they can be easily adapted through Theorem \ref{UL-bound} to obtain better bounds for $F'$. This observation leads to the following instantiation of $(\varepsilon,\delta)$-reconstruction-privacy based on the Chernoff bound.

\begin{corollary}[Testing $(\varepsilon,\delta)$-reconstruction-privacy] \label{test}
With the upper bounds $U_i$ and $L_i$ in Equations (2)-(5), for a micro group $g$,
for $\varepsilon \in (0,1+\frac{(1-p)/m}{pf}]$ and $\delta \in [0,1]$, $g^*$ is $(\varepsilon,\delta)$-reconstruction-private if and only if, for every $SA$ value in $g$,
\begin{equation}
\delta\leq min\{U_i(\theta,\mu),L_i(\theta,\mu)\}
\end{equation}
where $\theta=\frac{\varepsilon |g|pf}{\mu}$ and $\mu=|g|(fp+(1-p)/m)$.
\end{corollary}

The range $(0,1+\frac{(1-p)/m}{pf}]$ of $\varepsilon$ corresponds to the common range $(0,1]$ of $\theta$ for all of Equations (2)-(5). The condition in Equation (8) can be tested efficiently because all parameters in $\theta$ and $\mu$ are known to the data publisher.

\section{Achieving Privacy}
We now consider the second major question: how to achieve $(\varepsilon,\delta)$-reconstruction-privacy on the
published data $D^*$ for a given data set $D$. Corollary \ref{test} gives an efficient condition for $(\varepsilon,\delta)$-reconstruction-privacy, but it does not provide a clue on how to achieve this condition if it fails. As the first step towards an answer, we rewrite Equation (8) into a constraint on the size $|g|$ of a micro group $g$. Then we present an algorithm to enforce this constraint. Below, we consider $(L_1,U_1)$ and $(L_2,U_2)$ separately.

\begin{theorem} \label{simplified1}
With the upper bounds $U_1(\theta,\mu)$ and $L_1(\theta,\mu)$ in Equations (2) and (3), for a micro group $g$, $\varepsilon \in (0,1+\frac{(1-p)/m}{pf}]$, and $\delta \in [0,1]$, $g^*$ is $(\varepsilon,\delta)$-reconstruction-private if and only if, for the maximum frequency $f$ of any $SA$ value occurring in $g$,
\begin{equation}
|g| \leq  \frac{\ln \delta}{w \ln \left(\frac{e^{-\theta}}{(1-\theta)^{(1-\theta)}}\right)}
\end{equation}
where $w=fp+(1-p)/m$ and $\theta=\frac{\varepsilon pf}{w}$.
\end{theorem}
\begin{proof}
First, we show two claims. Let $X=\frac{e^{\theta}}{(1+\theta)^{(1+\theta)}}$ and $Y=\frac{e^{-\theta}}{(1-\theta)^{(1-\theta)}}$.

\emph{Claim 1: for $\theta \in (0,1]$, $X\geq Y$, thus, $min\{L_1,U_1\}=L_1$}. Note $\frac{X}{Y}$ approaches 1 as $\theta$ approaches 0. To show the claim, it suffices to show that $\frac{X}{Y}$ is non-decreasing, equivalently, the derivative of $\frac{X}{Y}$ wrt $\theta$ is non-negative for $\theta \in (0,1]$. Note
\[ \ln \frac{X}{Y}=2\theta +(1-\theta) \ln (1-\theta) -(1+\theta) \ln (1+\theta)\]
Differentiating both sides wrt $\theta$, we get
\[ \frac{Y}{X} (\frac{X}{Y})' =2 + [-\ln (1-\theta)+(1-\theta)\frac{-1}{1-\theta}] - [\ln (1+\theta)+(1+\theta)\frac{1}{1+\theta}] \]
and
\[ (\frac{X}{Y})' = -\frac{X}{Y} \ln (1-\theta^2) \geq 0 \]
The last inequality follows because $X$ and $Y$ are non-negative and $\theta$ is in $(0,1]$. This shows Claim 1.

\emph{Claim 2: for $\theta \in (0,1]$, $Y$ is in $(0,1)$ and is non-increasing}.
We show
that the derivative of $Y$ is non-positive (thus, $Y$ is non-increasing) for $\theta \in (0,1]$.  Then the claim follows from the fact that $Y$ approaches 1 as  $\theta$ approaches 0.
\[ \ln Y =-\theta -[(1-\theta) \ln (1-\theta)] \]
Differentiating both sides wrt $\theta$ gives
\[ Y' = Y [-1 -(-\ln (1-\theta)+(1-\theta) \frac{-1}{1-\theta}) = Y \ln (1-\theta) \leq 0\]
The last inequality follows because $Y$ is non-negative and $\theta$ is in $(0,1]$. This shows Claim 2.

From Claim 1, $L_1(\theta,\mu)\leq U_1(\theta,\mu)$, so Equation (8) degenerates into $\delta \leq L_1(\theta,\mu)=Y^{\mu}$, and $\ln \delta \leq \mu \ln Y = |g|w \ln Y$. From Claim 2, $Y$ is in $(0,1)$, so $\ln Y<0$, and $|g|\leq \frac{\ln \delta}{w \ln Y}$.
As $f$ increases, $w$ and $\theta=\frac{\varepsilon p}{p+\frac{1-p}{mf}}$ increase, and from Claim 2, $Y$ is in $(0,1)$ and is non-increasing, thus, $\ln Y$ is decreasing. Since both $\ln \delta$ and $\ln Y$ are negative, $\frac{\ln \delta}{w\ln Y}$ is minimized when $f$ is maximized. So Equation (8) degenerates into Equation (9).
\end{proof}

 Observe that the right-hand side of the condition in Equation (9) is a constant if the maximum frequency $f$ is kept unchanged. The idea of our algorithm to enforce this condition is reducing $|g|$ while keeping $f$ unchanged. The next theorem gives a similar rewriting based on the bounds $L_2$ and $U_2$ in Equations (4) and (5).

\begin{theorem} \label{simplified2}
With the upper bounds $U_2(\theta,\mu)$ and $L_2(\theta,\mu)$ in Equations (4) and (5), for a micro group $g$, $\varepsilon \in (0,1+\frac{(1-p)/m}{pf}]$, and $\delta \in [0,1]$, $g^*$ is $(\varepsilon,\delta)$-reconstruction-private if and only if, for the maximum frequency $f$ of any $SA$ value occurring in $g$,
\begin{equation}
|g|\leq  =\frac{-2\ln \delta}{w\theta^2} \label{size}
\end{equation}
where $w=fp+(1-p)/m$ and $\theta=\frac{\varepsilon pf}{w}$.
\end{theorem}
\begin{proof}
For $\theta \geq 0$, $L_2(\theta,\mu)\leq U_2(\theta,\mu)$, so Equation (8) degenerates into $\delta \leq L_2(\theta,\mu)$, where $\theta=\frac{\varepsilon |g|pf}{\mu}$ and $\mu=|g|w$. Note $\theta=\frac{\varepsilon pf}{w} = \frac{\varepsilon p}{p+\frac{1-p}{mf}}$. As $f$ increases, $\theta$ and $\mu=|g|w$ increase, hence, $L_2(\theta,\mu)=exp(-\frac{\theta^2}{2} \mu)$ decreases.
Therefore, it suffices to consider the maximum frequency $f$ in $g$ for checking
$\delta\leq L_2$. The rest of the proof follows from the following rewriting:
%\begin{eqnarray*}
% \delta \leq exp(-\frac{\theta^2}{2}\mu) & \Leftrightarrow & exp(\frac{\theta^2}{2}\mu)\leq 1/\delta \\
% & \Leftrightarrow & \frac{\theta^2}{2}\mu \leq -\ln \delta \\
% & \Leftrightarrow & \frac{1}{2}(\frac{(\varepsilon |g| p f^{max}}{\mu})^2 \mu \leq -\ln
% \delta\\
%  & \Leftrightarrow & \frac{1}{2}\frac{(\varepsilon p f^{max})^2|g|}{w} \leq -\ln
% \delta\\
% & \Leftrightarrow & |g|\leq \frac{-2w \ln \delta}{(\varepsilon p f^{max})^2}
%\end{eqnarray*}
\[
 \delta \leq exp(-\frac{\theta^2}{2}\mu) \Leftrightarrow \mu \leq -\frac{2\ln \delta}{\theta^2}
% exp(\frac{\theta^2}{2}\mu)\leq 1/\delta \\
% & \Leftrightarrow & \frac{\theta^2}{2}\mu \leq -\ln \delta \\
% & \Leftrightarrow & \mu \leq -\frac{2\ln \delta}{\theta^2}\\
 \Leftrightarrow  |g|\leq \frac{-2\ln \delta}{w \theta^2} \]
%\end{eqnarray*}
\end{proof}

 In the rest of this section, we develop an algorithm for achieving $(\varepsilon,\delta)$-reconstruction-privacy based on Theorem \ref{simplified2}, but a similar algorithm can be developed based on Theorem \ref{simplified1}. According to Theorem \ref{simplified2}, if $|g|\leq  s_g$ fails, where $s_g=\frac{-2\ln \delta}{w\theta^2}$, $g^*$ is not $(\varepsilon,\delta)$-reconstruction-private. There are several options to restore this inequality. One option is increasing $s_g$ by reducing either the retention probability $p$ or the maximum frequency $f$ in $g$. Another option is decreasing $|g|$ by discarding some records. None of these options is desirable because they either make the data set more random or distort the global data distribution.

Our observation is that $|g|$ in Equation (10) really refers to the number of \emph{independent} Poisson trials in the randomization process for generating $g^*$. This can be seen from $\mu=E[O^*]=|g|(fp+(1-p)/m)$ (Lemma \ref{f'}(i)) where $|g|$ is the number of indicator variables $X_k$ for the event that the $k$-th row in $g^*$ has a particular $SA$ value $x$ (see the proof of Lemma \ref{f'}).
Since the upper bounds in Equations (2)-(5) decrease exponentially in $\mu$,
reducing $|g|$ is highly effective to increase these upper bounds, which helps restore the inequality in Equation (10), provided that the frequency $f$ remains unchanged. At the same time, we want to preserve the frequency of each $SA$ value to minimize the distortion to data distribution. To meet both requirements, we shall randomize a sample $g_1$ of $g$ and scale the randomized data $g_1^*$ back to the original size $|g|$. The key is to preserve the frequency of each $SA$ value in both sampling and scaling operations. This task is performed by the following three functions. Assume $|g|> s_g$.

\begin{enumerate}
\item $Sampling(g,s_{g})$: this function takes a sample of the size $s_g$ from $g$ such that the number of records for each $SA$ value is reduced by the same fraction. Let $b=s_g/|g|$ (note $b<1$). For each $SA$ value $x$ occurring in $g$, let $g_1$ contain any $\lfloor |g_x| b\rfloor$ records from $g_x$ and one additional record from $g_x$ with probability $|g_x|b-\lfloor|g_x |b\rfloor$, where $g_x$ denotes the set of records in $g$ for $x$. Note that all records in $g_x$ are identical. Return $g_1$.
    This step reduces the number of independent trials to $s_g$ while preserving the frequency of each $SA$ value.
   % Note that
%    $|g_1|$ may not be exactly $s_g$ due to the coin flip of the additional record, but the mean of $|g_1|$ is equal to $s_g$.

\item $Perturbing(g_1,p,m)$: this function randomizes the $SA$ values of the records in $g_1$
as described in Section 3.1 and returns the randomized $g_1^*$.

\item $Scaling(g_1^*,|g|)$: this function scales up $g_1^*$ to the original size $|g|$ while preserving the frequency of each $SA$ value. Let $b'=|g|/|g_1^*|$. For each record
$r^*$ in $g_1^*$, let $g_2^*$ contain $\lfloor b'\rfloor$ duplicates of
$r^*$ and one additional duplicate of $r^*$ with probability
$b'-\lfloor b'\rfloor$. Return $g_2^*$.
Note that the duplication does not increase the number of independent trails because all duplicates of $t^*$ originate from the same independent trial for $r^*$.
\end{enumerate}

The algorithm based on the above idea is described in Algorithm \ref{algorithm1}. The input consists of
$D,p,m,\varepsilon,\delta$ and the output is $D_2^*$. For each micro group $g$, if $|g|\leq s_{g}$, $g_2^*$ is equal to $g^*$. Otherwise, $g_2^*$ is produced by the three steps on Lines 7-9 described above. $D_2^*$ contains all $g_2^*$.

\begin{example}\label{example4}
Suppose that a micro group $g$ contains 5 records for $x_1$ and 15 records for $x_2$. $|g|=20$, $|g_{x_1}|=5$, $|g_{x_2}|=15$. Assume $s_g=15$. Since $|g|> s_g$, $Sampling(g,s_{g})$
produces a sample $g_1$ of $g$ as follows. $b=s_g/|g|=0.75$.
$g_1$ contains $\lfloor 5 \times 0.75\rfloor=3$ records from $g_{x_1}$ and one additional record from $g_{x_1}$ with probability $5\times
0.75-3=75\%$; $g_1$ contains $\lfloor 15 \times 0.75\rfloor=11$
records from $g_{x_2}$ and one additional record from $g_{x_2}$ with
probability $15\times 0.75-11=25\%$. Suppose that after coin flips,
$g_1$ contains 4 records from $g_{x_1}$ and 11 records from $g_{x_2}$.
$Perturbing(g_1,p,m)$ produces the randomized version of $g_1$, $g_1^*$.

$Scaling(g_1^*,|g|)$ scales up $g_1^*$ to the size $|g|$ as follows.
$b'=|g|/|g_1^*|=20/15=1.33$. For each record $r^*$ in $g_1^*$,
$g_2^*$ contains $\lfloor b'\rfloor=1$ duplicate of $r^*$ and contains one additional duplicate with probability $1.33-1=33\%$. Suppose that after coin flips, one additional duplicate for $x_1$ is chosen, and four additional duplicates for $x_2$ are chosen. So $g_2^*$
contains 5 records for $x_1$ and 15 records for $x_2$. In general, $|g_2^*|$ may not be exactly equal to $|g|$.
\end{example}

\begin{algorithm}[h!]
\caption{Achieving Reconstruction Privacy} \label{algorithm1}
Input: $D,p,m,\varepsilon,\delta$\\
Output: Randomized $D_2^*$ that is
$(\epsilon,\delta)$-reconstruction-private
\begin{algorithmic} [1]
 \STATE $D_2^*\leftarrow {\O}$
 \FORALL {micro groups $g$ in $D$}
    \STATE compute $s_g=\frac{-2\ln \delta}{w\theta^2}$ using Equation (\ref{size})
    \IF {$|g|\leq s_g$}
       \STATE $g_2^* \leftarrow Perturbing(g,p,m)$
     \ELSE
       \STATE $g_1 \leftarrow Sampling(g,s_g)$
       \STATE $g_1^* \leftarrow Perturbing(g_1,p,m)$
       \STATE $g_2^*\leftarrow Scaling(g_1^*,|g|)$
    \ENDIF
    \STATE add $g_2^*$ to $D_2^*$
 \ENDFOR
 \STATE return $D_2^*$
\end{algorithmic}

\vspace{0.5cm}

$Sampling(g,s_{g})$:
\begin{algorithmic} [1]
 \STATE $temp \leftarrow{\O}$
 \STATE $b\leftarrow s_g/|g|$
 \FORALL {$SA$ value $x$ occurring in $g$}
   \STATE $g_x \leftarrow$ the set of records in $g$ having $x$
   \STATE add to $temp$ any $\lfloor |g_x |b\rfloor$ records from $g_x$
   \STATE add to $temp$ one additional record from $g_x$ with probability
 $|g_x|b-\lfloor|g_x |b\rfloor$
 \ENDFOR
 \STATE return $temp$
\end{algorithmic}

\vspace{0.5cm}

$Perturbing(g_1,p,m)$:
\begin{algorithmic} [1]
 \STATE $temp \leftarrow {\O}$
 \FORALL {record $r$ in $g_1$}
   \STATE let $r^*$ be $r$ with $SA$ perturbed with retention
   probability $p$
   \STATE add $r^*$ to $temp$
 \ENDFOR
 \STATE return $temp$
\end{algorithmic}

\vspace{0.5cm}

$Scaling(g_1^*,|g|)$:
\begin{algorithmic}[1]
 \STATE $b' \leftarrow |g|/|g_1^*|$
 \STATE $temp \leftarrow {\O}$
 \FORALL {record $r^*$ in $g_1^*$}
  \STATE add to $temp$ $\lfloor b'\rfloor$ duplicates of $r^*$
  \STATE add to $temp$ one additional duplicate of $r^*$ with probability
$b'-\lfloor b'\rfloor$
 \ENDFOR
 \STATE return $temp$
\end{algorithmic}
\end{algorithm}

We show that $D_2^*$ produced by Algorithm \ref{algorithm1} satisfies some interesting properties with respect to privacy and utility. Consider a micro group $g$ such that $|g|> s_g$. Let $g_1,g_1^*,g_2^*$ be computed for $g$ in Algorithm \ref{algorithm1}, and let $O_g^*,O_{g_1}^*,O_{g_2}^*$ be the observed count of a particular $SA$ value $x$ in $g^*,g_1^*,g_2^*$. Let $f_g$ and $f_{g_1}$ be the frequency of $x$ in $g$ and $g_1$. Let $F'_g,F'_{g_1},F'_{g_2}$ be the MLEs reconstructed from $g^*,g_1^*,g_2^*$. $u \simeq v$ denotes that $u$ and $v$ are equal modulo the coin flips in Scaling and Sampling. It is easy to see a few simple facts:

\begin{itemize}
\item Fact 1: $f_{g_1} \simeq f_g$, that is, Sampling preserves the frequency of $x$ in $g$. This is because the count of every $x$ in $g$ is reduced by the \emph{same} factor $b$ modulo the coin flips.

\item Fact 2: $O_{g_2}^*/|g_2^*| \simeq O_{g_1}^*/|g_1^*|$, that is, Scaling preserves the frequency of $x$ in $g_1^*$. This is because each record in $g_1^*$ is duplicated $b'$ times modulo the coin flips.

\item Fact 3: $F'_{g_1} \simeq F'_{g_2}$, that is, $g_1^*$ and $g_2^*$ give the same estimate of $f_g$. This follows from $F'_{g_i}= \frac{O_{g_i}^*/|g_i^*| -
(1-p)/m}{p}$, $i=1,2$ (Lemma \ref{f'}(ii)) and Fact 2.

\item Fact 4: $E[O_{g_2}^*] \simeq E[O_g^*]$ and $|g^*| \simeq |g_2^*|$. $|g^*| \simeq |g_2^*|$ follows from Sampling and Scaling. From Lemma \ref{f'}(i), $E[O_g^*]=|g| (fp +(1-p)/m)$ and $E[O_{g_1}^*]=s_g (f_1 p +(1-p)/m)$. Since Scaling  duplicates each $x$ occurrence in $g_1^*$ $\frac{|g|}{s_g}$ times, $E[O_{g_2}^*] \simeq \frac{|g|}{s_g} E[O_{g_1}^*] = |g| (f_1 p +(1-p)/m)$. Then Fact 1 implies $E[O_{g_2}^*] \simeq E[O_g^*]$.
\end{itemize}

\begin{theorem}[Privacy]\label{th:privacy}
For each micro group $g$, $g_2^*$ is
$(\varepsilon,\delta)$-reconstruction-private.
% (modulo the coin flips in Sampling and Scaling).
\end{theorem}
\begin{proof}
If $|g|\leq s_g$, $g_2^*$ is $(\varepsilon,\delta)$-reconstruction-private (Theorem \ref{simplified2}).  We assume $|g|> s_g$. $g_1^*$ is $(\varepsilon,\delta)$-reconstruction-private because $|g_1| \simeq s_{g_1}$
(Theorem \ref{simplified2}). $|g_1| \simeq s_{g_1}$ follows because $f_{g_1} \simeq f_g$ (Fact 1) implies $s_g \simeq s_{g_1}$, and from $|g_1| \simeq s_g$, $|g_1| \simeq s_{g_1}$.
Facts 1 and 3 imply $\frac{F'_{g_2}-f_g}{f_g} \simeq \frac{F'_{g_1}-f_{g_1}}{f_1}$.
%, that is, the estimate of $f_g$ based on $g_2^*$
%has the same error as the estimation of $f_{g_1}$ based on $g_1^*$.
So, $\Pr[\frac{F'_{g_2}-f_g}{f_g} > \varepsilon]\simeq\Pr[\frac{F'_{g_1}-f_{g_1}}{f_{g_1}}> \varepsilon]$, and $\Pr[\frac{F'_{g_2}-f_g}{f_g} < - \varepsilon]\simeq\Pr[\frac{F'_{g_1}-f_{g_1}}{f_{g_1}} < - \varepsilon]$.
Since $g_1^*$ is $(\varepsilon,\delta)$-reconstruction-private, so is $g_2^*$.
\end{proof}

Below, we show that $F'_2$ has the same mean as $F'$. Let $S$ be any set of micro groups, and let $S^*$ and $S_2^*$ be the sets of corresponding records in $D^*$ and $D_2^*$, respectively. For any $SA$ value $x$, let $F'_2$ denote the estimated frequency of $x$ in $S$ based on $S_2^*$ and let $F'$ denote the estimated frequency of $x$ in $S$ based on $S^*$.

\begin{theorem}[Utility]\label{frequency}
$E[F'_2]\simeq E[F']$.
\end{theorem}
\begin{proof}
Let $O_2^*=\sum_{g\in S} O_{g_2}^*$ and $O^*=\sum_{g\in S} O_{g}^*$. Let $|S^*|= \sum_{g\in S} |g^*|$ and $|S_2^*|= \sum_{g \in S} |g_2^*|$. From Lemma \ref{f'}(ii),
$E[F']=\frac{E[O^*]/|S^*| -
(1-p)/m}{p}$ and $E[F'_2]=\frac{E[O_2^*]/|S_2^*| - (1-p)/m}{p}$.
From Fact 4, $|S^*|\simeq |S_2^*|$ and
$E[O^*] \simeq E[O_2^*]$, which implies $E[F'] \simeq E[F'_2]$.
\end{proof}

%\begin{remark}\label{remark3}

Despite $E[F'_2]\simeq E[F']$, $F'_2$ will have a larger error than $F'$ due to the reduced number of independent trials for $S_2^*$. This is exactly what we want in order to restore $(\varepsilon,\delta)$-reconstruction-privacy. However, the error increase for aggregate reconstruction is smaller than that for micro reconstruction because aggregation reconstruction involves more than one micro group.
%For example, in Example \ref{example3},
%the answer for the query on the predicate ``$Gender=M \wedge Disease=Cancer$" is reconstructed from the set of records in $D^*$ having $Gender=M$, which is the union of all micro groups of the form $Gender=M \wedge Job=a$, where $a$ is any job. Even though the number of independent trials in each of these micro groups is reduced, the total number of independent trials in these groups may still be large enough for accurate reconstruction.
%This is why the proposed method has more error-increasing effect on micro reconstruction than on legitimate aggregate queries and data analysis.
We will evaluate this claim empirically in Section 6.

\section{Empirical Evaluation}
This empirical study aims to answer two questions: The first question is ``to what extent is $(\varepsilon,\delta)$-reconstruction-privacy violated assuming that major privacy definitions are satisfied?". The second question is ``what price will be paid for having $(\varepsilon,\delta)$-reconstruction-privacy?" Section \ref{eSetup} introduces our data sets and utility metrics. Section 6.2 presents the findings in the data publishing setting and Section 6.3 presents the findings in the output perturbation setting.

\begin{figure}[h]
\subfigure{
\begin{minipage}[t]{0.45\linewidth}
\centering
\includegraphics[width=4cm,height=3cm]{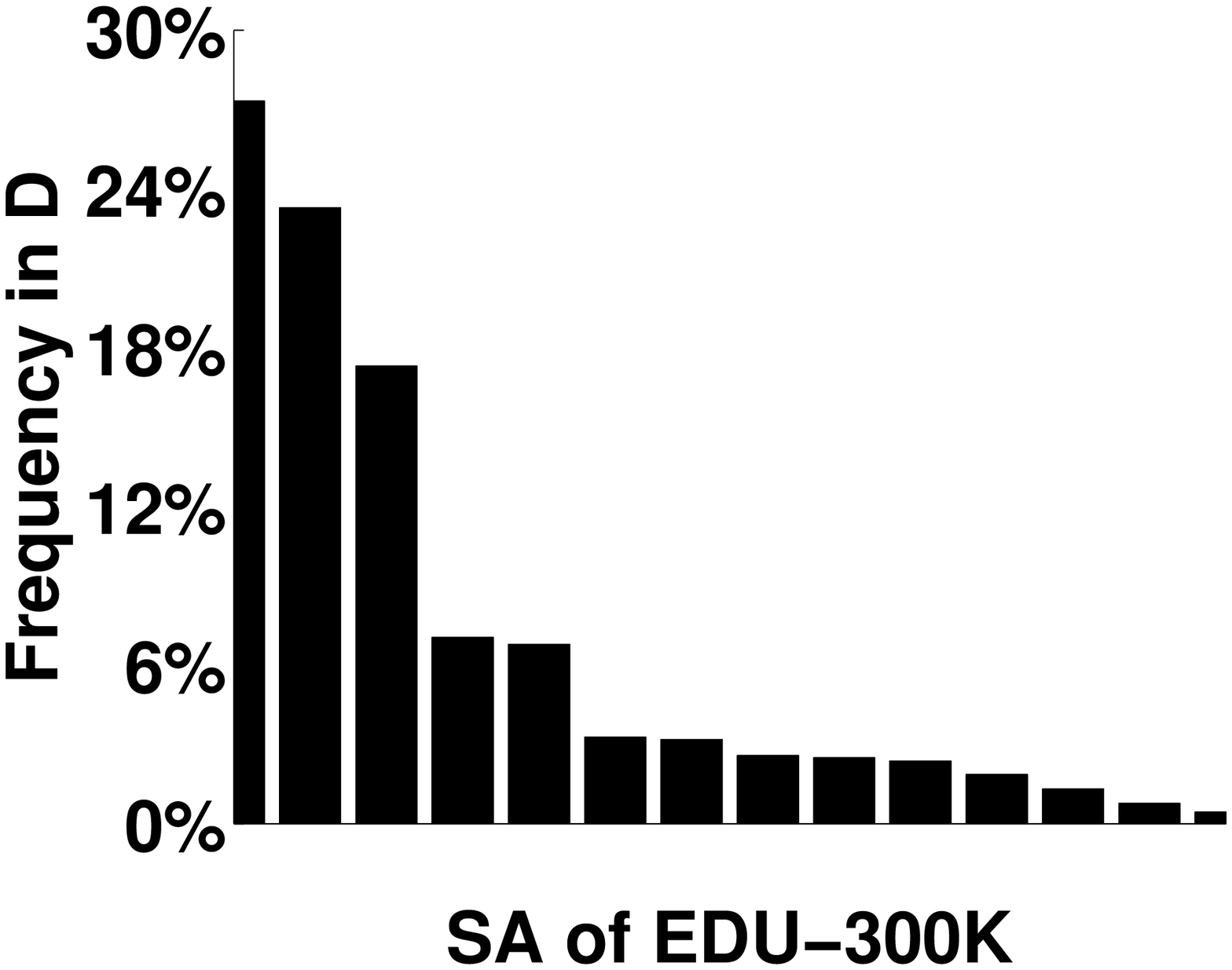}
%% old a and b values for F'
\end{minipage}}
\hfill
\subfigure{
\begin{minipage}[t]{0.45\linewidth}
\centering
\includegraphics[width=4cm,height=3cm]{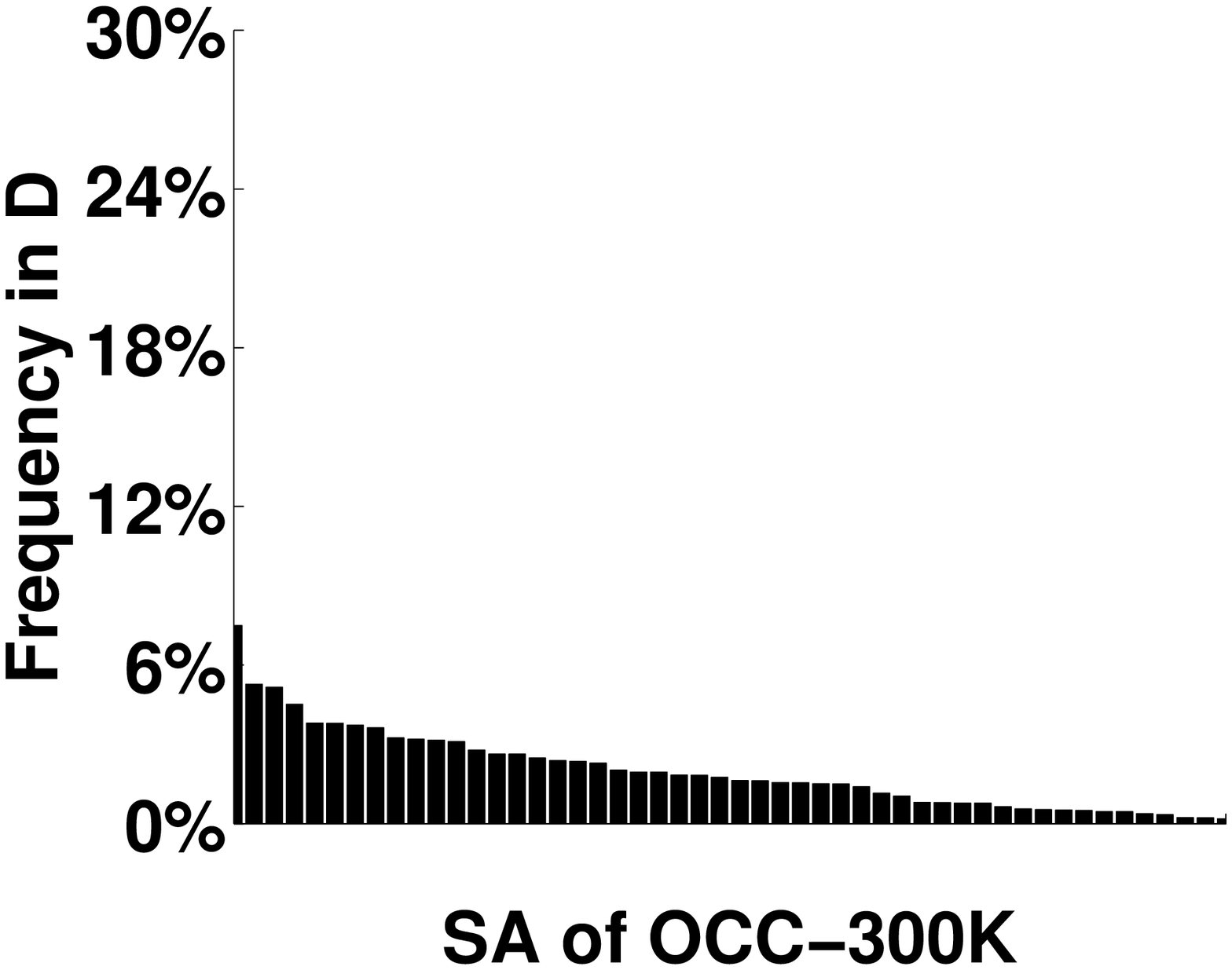}
\end{minipage}}
\caption{Frequency Distribution for $SA$}
%coefficient $\theta$ (x-axis)} \label{lvsa}
\label{figure:distribution}
\end{figure}

\subsection{Experimental Setup}\label{eSetup}

\textbf{Data Sets}. We utilize the real
CENSUS data containing personal information of 500K American adults, previously used in \cite{XT06b},\cite{MKG+006}, and \cite{CW10}.
Table \ref{table:attribues} shows the 7 discrete attributes of the data. Two base tables were generated from CENSUS. OCC denotes the base table with Occupation as the sensitive attribute ($SA$) and the remaining attributes as the non-sensitive attributes ($NA$). EDU denotes the base table with Education as the sensitive attribute ($SA$) and the remaining attributes as the non-sensitive attributes ($NA$). OCC-$n$ and EDU-$n$ denote the samples of cardinality $n$, where $n=100K,200K,300K,400K,500K$. Figure \ref{figure:distribution} shows the frequency distribution of $SA$ for OCC-300K and EDU-300K. EDU-300K has a more skewed distribution than OCC-300K.

\begin{table}[h]
\centering
\caption{Number of Values in Attributes}
\begin{tabular}{|l|l|} \hline
\textbf{Attributes} & \textbf{Domain Size}\\ \hline
Age & 77\\ \hline
Gender & 2\\ \hline
Education & 14\\ \hline
Marital & 6\\ \hline
Race & 9\\ \hline
Work-class & 7\\ \hline
Occupation & 50\\ \hline
\end{tabular}
\label{table:attribues}
\end{table}

\textbf{Count Queries}. We evaluate the utility of data analysis through count queries of the following form
\begin{equation}
\begin{split}
 &SELECT\ COUNT\ (*)\ FROM\ D \\
 &WHERE\ A_1=a_1 \wedge \cdots \wedge A_d = a_d \wedge SA = x_i \\
\end{split}
\label{equation:queryForm}
 \end{equation}
where $\{A_1,...,A_d\}$ is a subset of non-sensitive attributes and $a_j$ is a value from the domain of $A_j$, $j=1,\ldots,d$, and $x_i$ is a value from the domain of $SA$. The answer to the query, denoted by $ans$, is the count of records in $D$ that satisfy the predicate in the WHERE clause. Since our primary
interest is in aggregate information, we consider only queries that have at least $0.1\%$ selectivity, where the \emph{selectivity} is defined as $ans/|D|$. This means that $d$ is restricted to be 1, 2, or 3 because a query for any larger $d$ has a selectivity less than $0.1\%$.
We generate a pool of 5,000 queries as follows. For each a query, we randomly select $d$ from $\{1,\ 2,\ 3\}$ with equal probability and randomly select $d$ non-sensitive attributes without replacement. For each attribute $A_i$ selected, we randomly choose a value $a_i$ from  the domain of $A_i$. Finally, we randomly choose a value $x_i$ from the domain of $SA$ and create a query following the template in Equation (\ref{equation:queryForm}). If the query has a selectivity of $0.1\%$ or more, we add it to the pool. This process is repeated until the pool contains 5,000 queries.

% equation (7) where {A_1,…,A_d} is a subset of non-sensitive attributes and a_j is a value from the domain of A_j,j=1,…,d, and x_i is a SA-value. Like in PP's method we generate a random query pool of count queries as follows. First, we created 200 random conditions of the form: A_1=a_1  AND…AND A_d=a_d. We randomly select a value d from {1, 2, 3} (with equal probability), randomly sample d non-sensitive attributes A_1,…,A_d without replacement, and for each A_i, we randomly select a value a_i from A_i's domain. Then, for each of these 200 conditions, and for each of the m values x_i in the domain of SA, we generated a count query following the template in equation (7).

\begin{table}[h]
\centering
\caption{Parameter Table}
\begin{tabular}{|l|l|} \hline
\textbf{Parameters} & \textbf{Settings}\\ \hline
$p$ & 0.1, 0.3, \textbf{0.5}, 0.7, 0.9\\ \hline
$\varepsilon$  & 0.1, 0.3, \textbf{0.5}, 0.7, 0.9\\ \hline
$\delta$  & 0.14, 0.22, \textbf{0.3}, 0.38, 0.46\\ \hline
$|D|$ & 100K, 200K, \textbf{300K}, 400K, 500K\\ \hline
\end{tabular}
\label{table:parameters}
\end{table}

\begin{figure*}[t!]
  \begin{center}
    \mbox{
      \subfigure[vs. $p$]{
        \includegraphics[width=4cm,height=3cm]{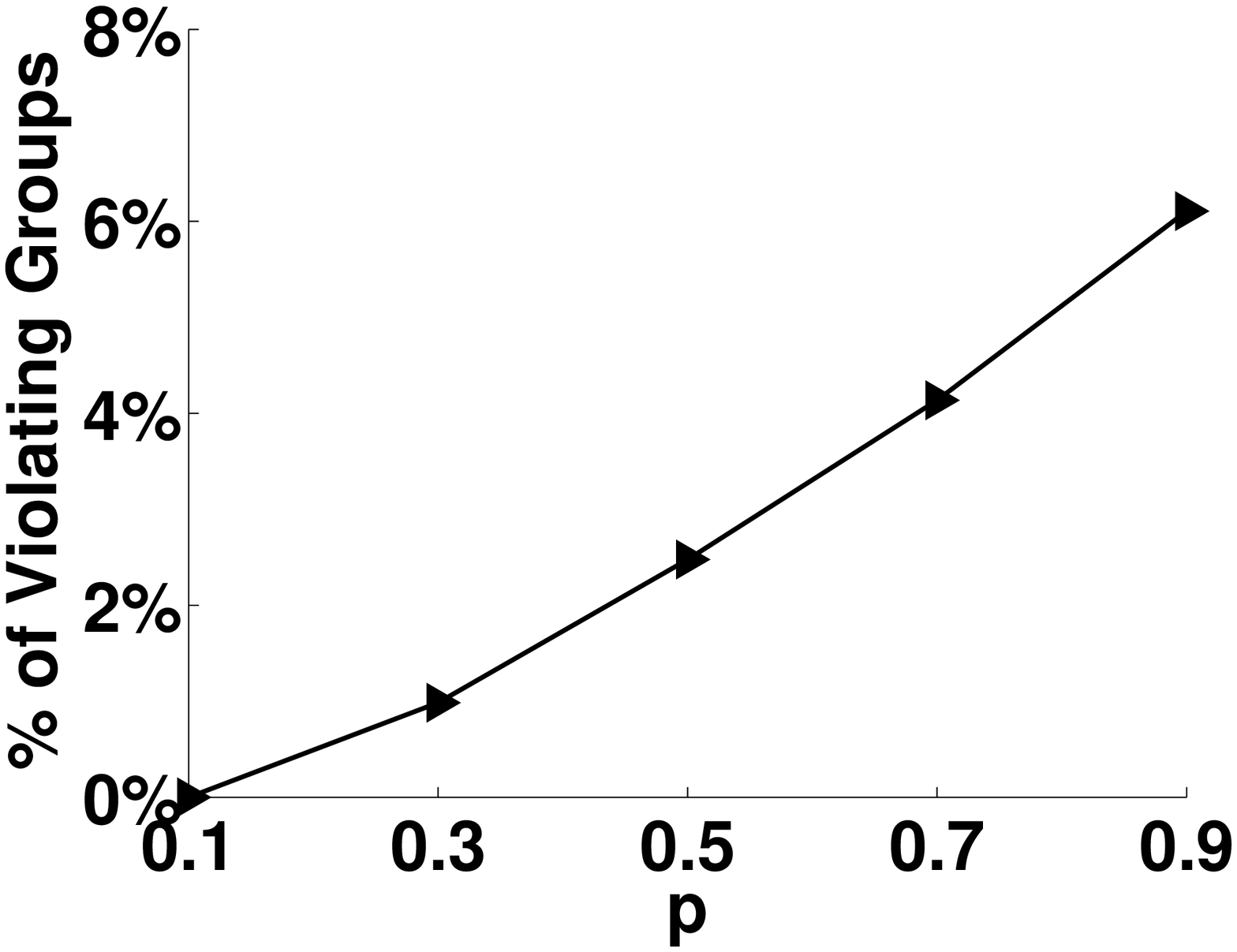}
      }
      \subfigure[vs. $\varepsilon$]{
        \includegraphics[width=4cm,height=3cm]{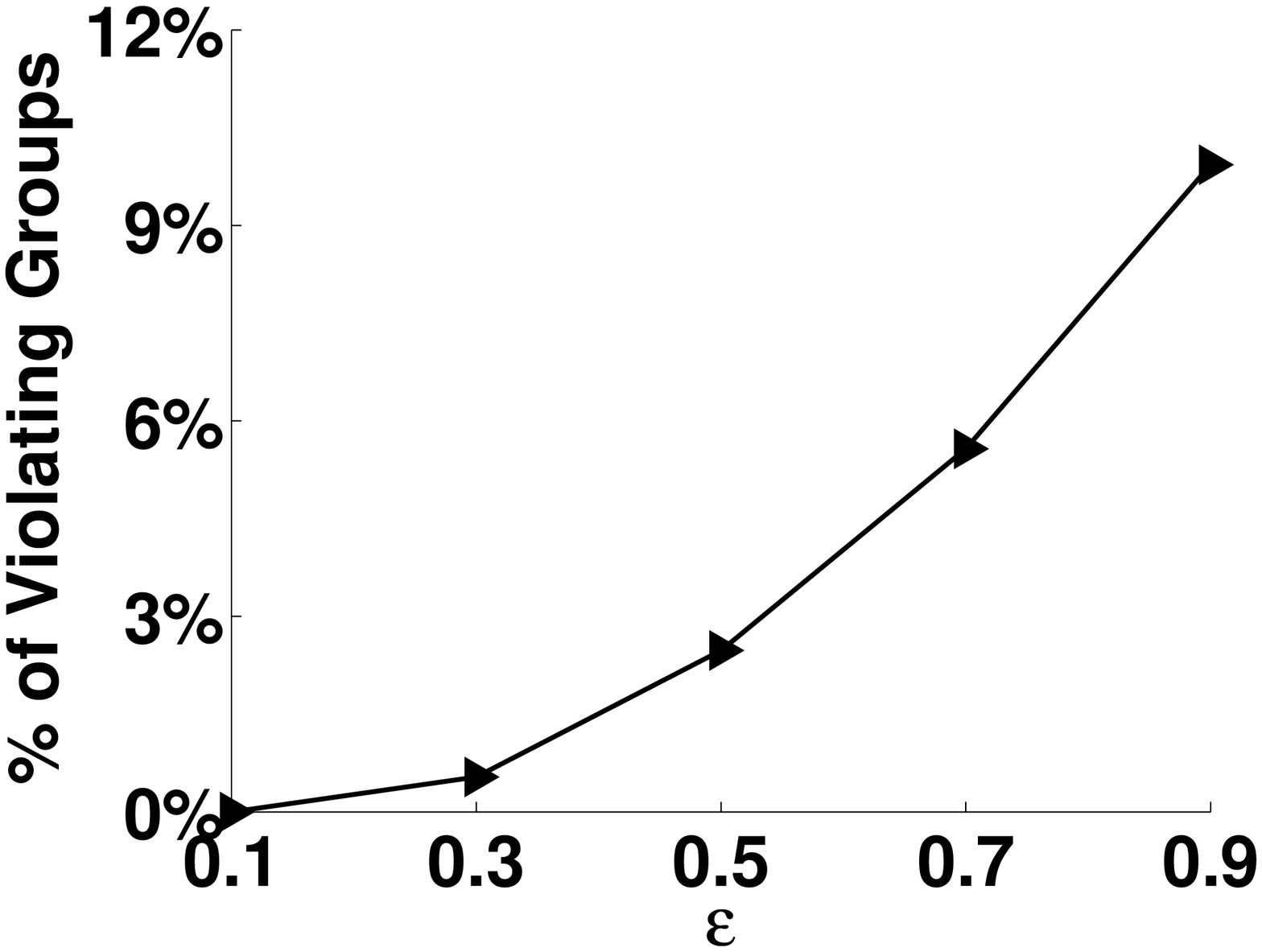}
        }
        \subfigure[vs. $\delta$]{
        \includegraphics[width=4cm,height=3cm]{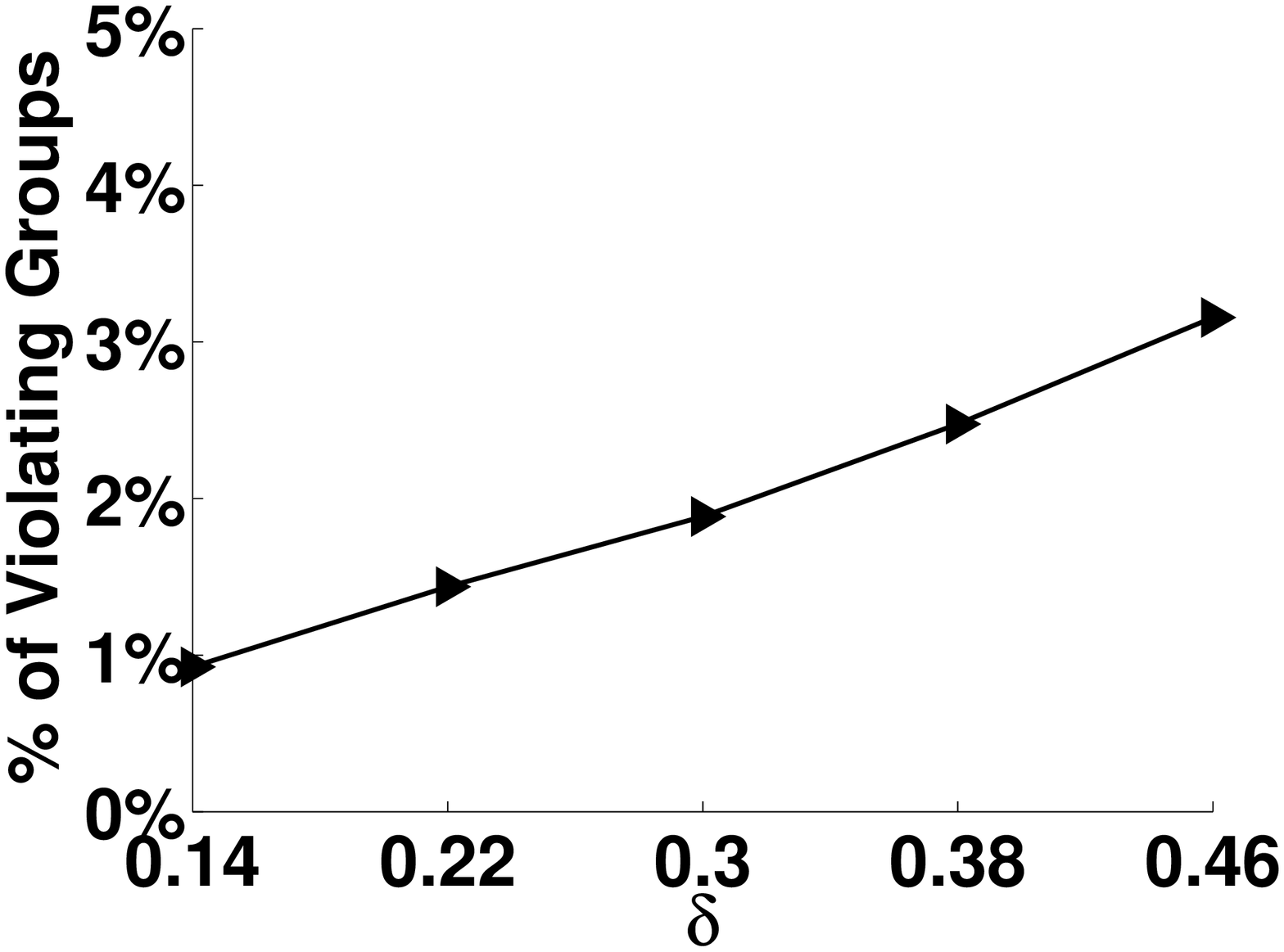}
         }
      }
      \subfigure[vs. $|D|$]{
        \includegraphics[width=4cm,height=3cm]{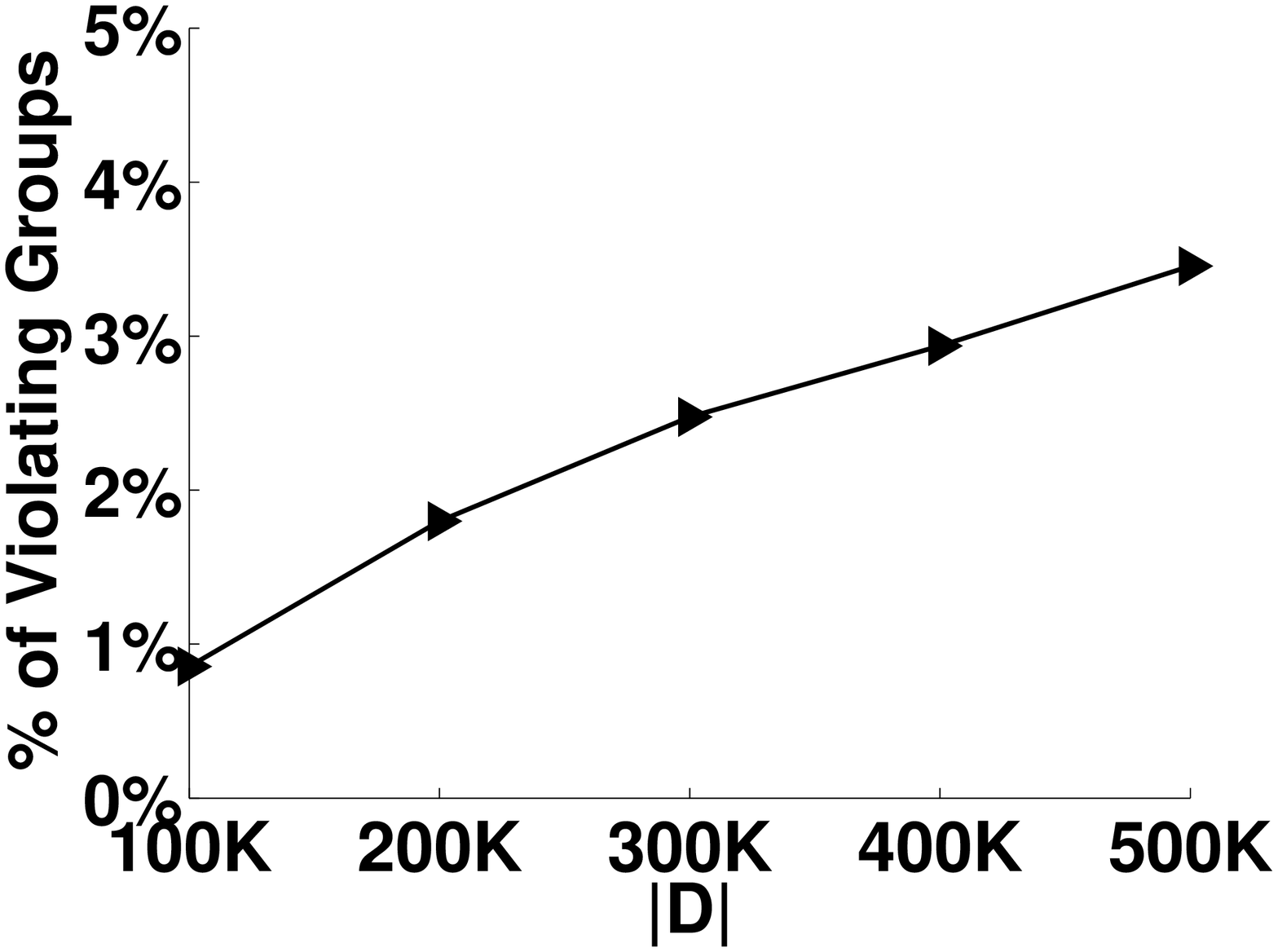}
      }
    \caption{EDU: \% of Violating Micro Groups in $D^*$}
    \label{fig:perEDU}
    \vspace{1cm}
     \mbox{
      \subfigure[vs. $p$]{
        \includegraphics[width=4cm,height=3cm]{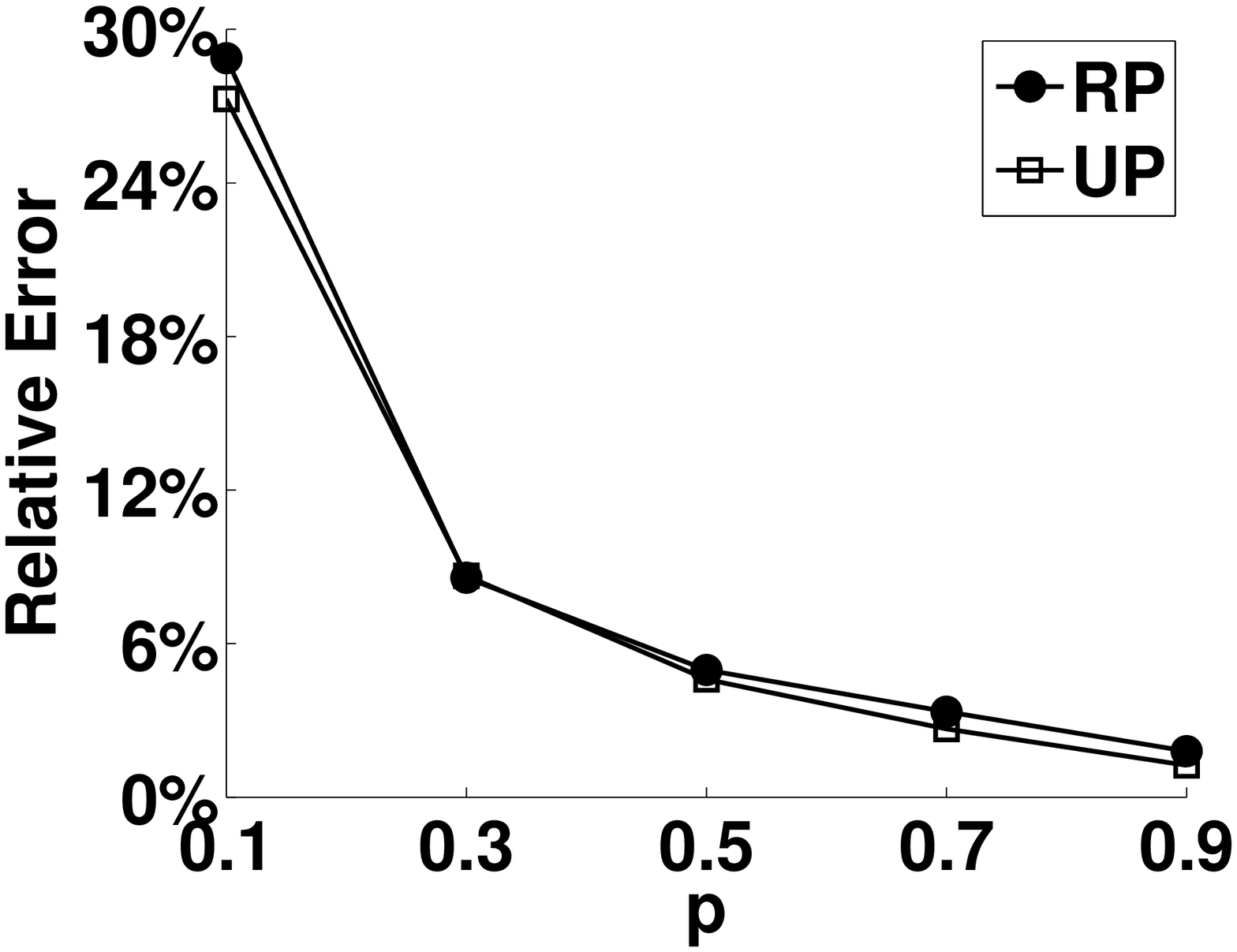}
      }
      \subfigure[vs. $\varepsilon$]{
        \includegraphics[width=4cm,height=3cm]{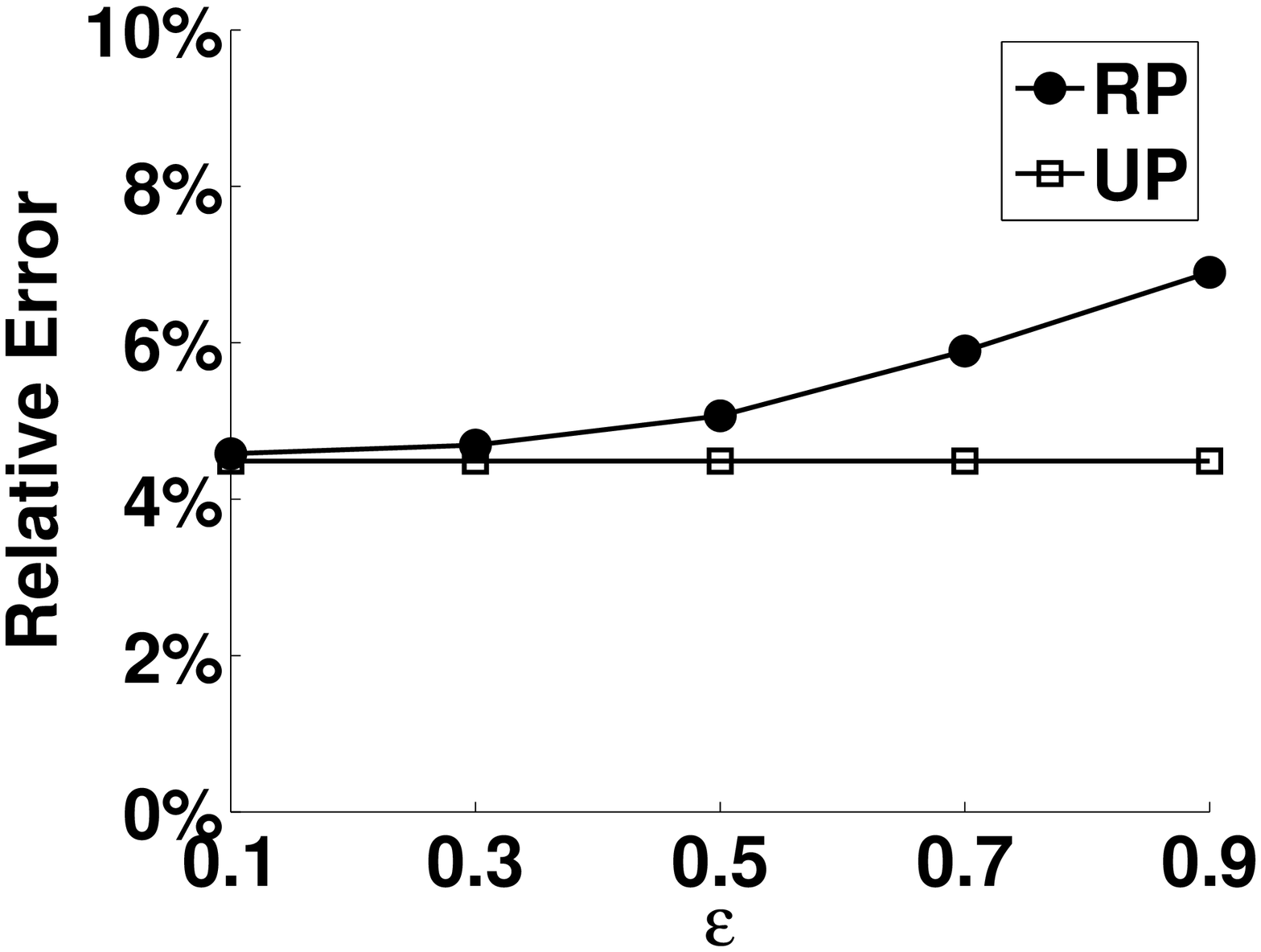}
        }
        \subfigure[vs. $\delta$]{
        \includegraphics[width=4cm,height=3cm]{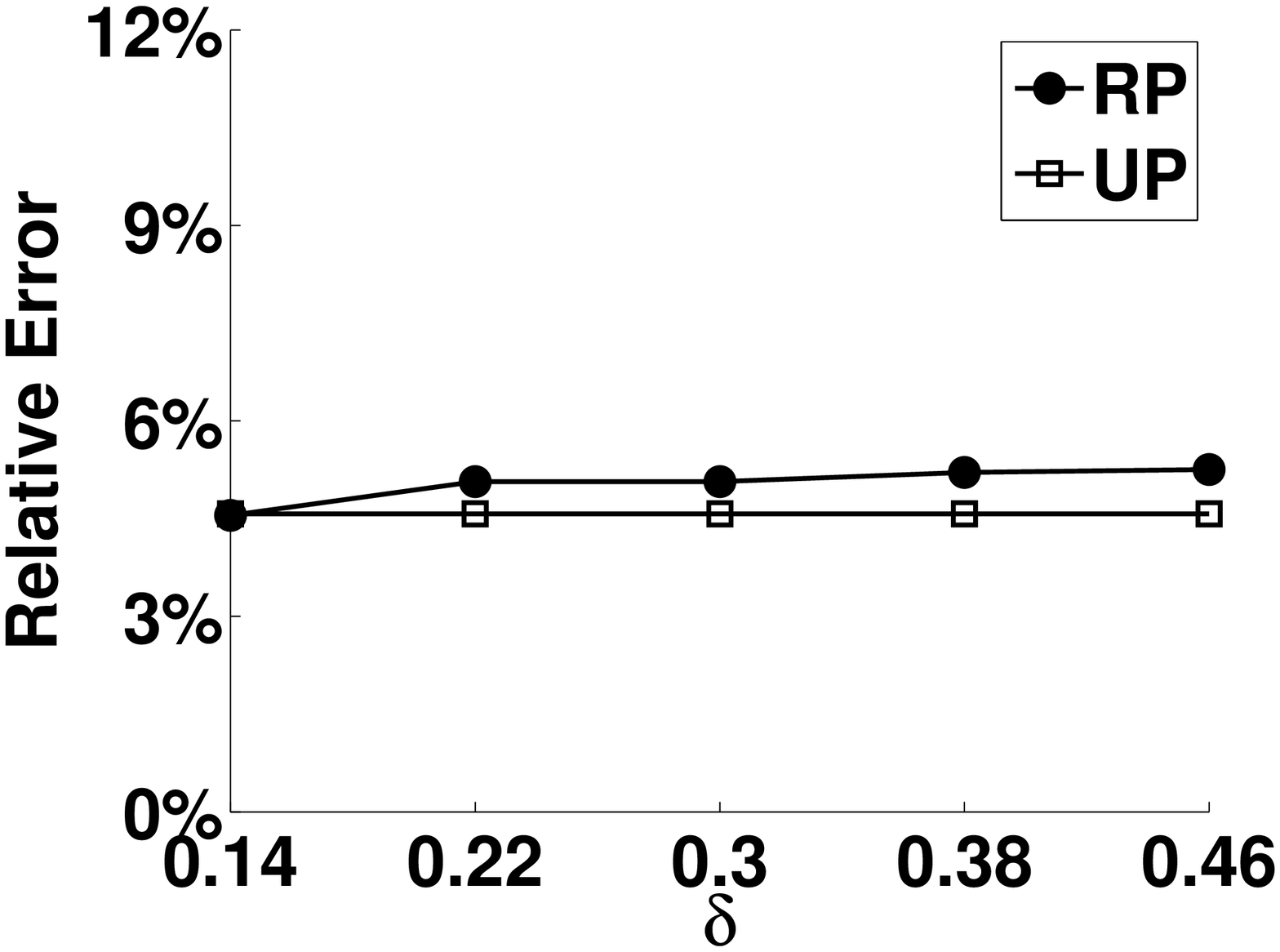}
      }
      }
      \subfigure[vs. $|D|$]{
        \includegraphics[width=4cm,height=3cm]{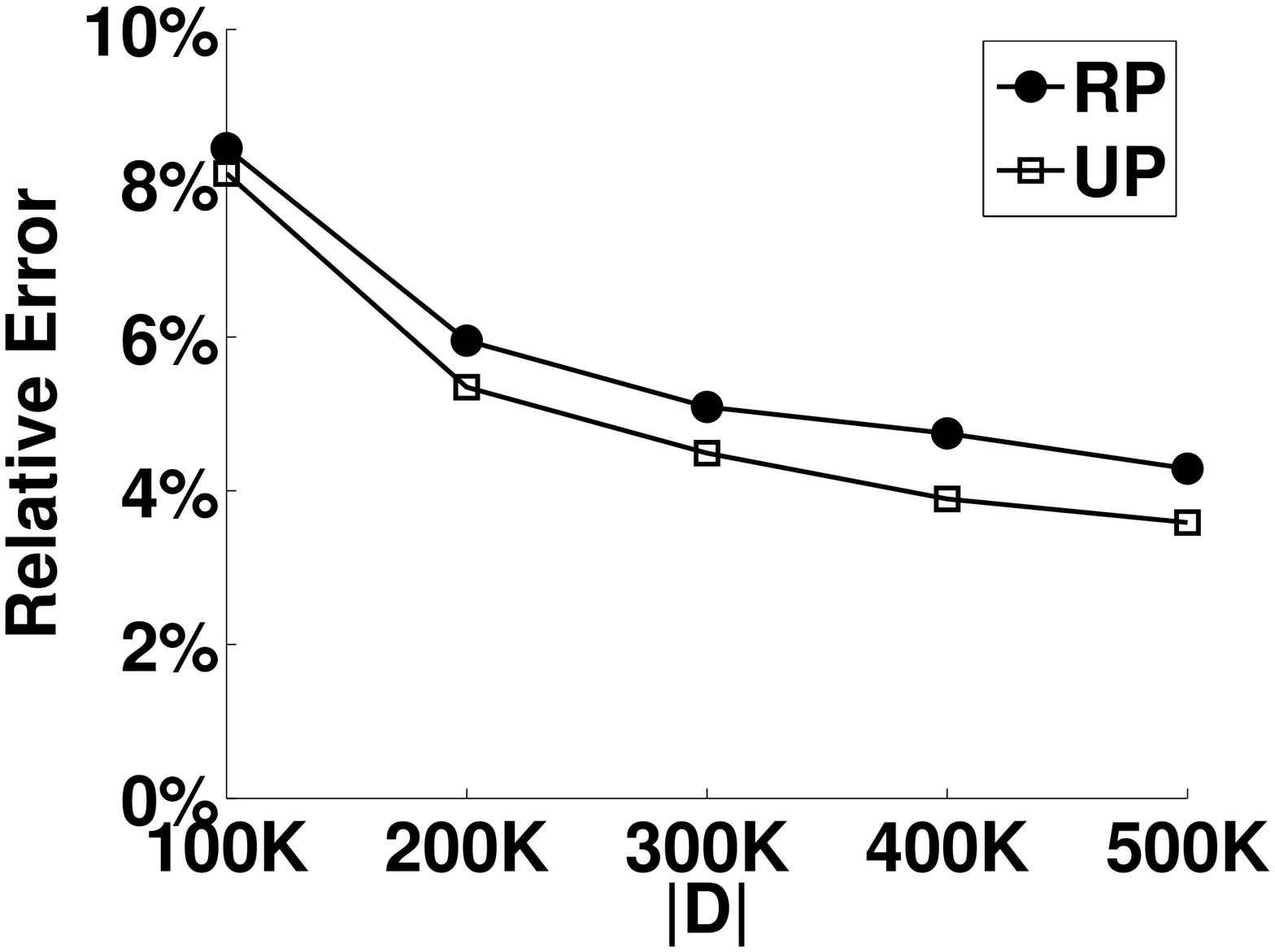}
      }
    \caption{EDU: Comparison of Relative Error for Count Queries}
    \label{fig:errEDU}
  \end{center}
\end{figure*}

\subsection{Findings in Data Publishing}
In the data publishing scenario, the randomized data $D^*$ is published and a query is answered using $D^*$ and a reconstruction process as described in Section 3.1.
Our study focuses on two questions: (i) To what extent is $(\varepsilon,\delta)$-reconstruction-privacy violated on $D^*$? (ii) What additional price is incurred for the protection of $(\varepsilon,\delta)$-reconstruction-privacy?
To answer the first question, we study the \emph{percentage of violating micro groups} in $D^*$ that fail to satisfy $(\varepsilon,\delta)$-reconstruction-privacy.
To answer the second question, we measure the (average) \emph{relative error} for answering the count queries in our query pool, defined as $\frac{|est-ans|}{ans}$, where $ans$ is the true answer and $est$ is the estimated answer. We compare the relative error generated using $D_2^*$ produced by Algorithm 1, denoted by RP (for reconstruction privacy), with the relative error generated using $D^*$ produced by the standard uniform perturbation, denoted by UP. Both methods use a retention probability $p$ to randomize the data, thus, ensure some uncertainty of the $SA$ value in a record such as $\rho_1$-$\rho_2$ privacy. According to \cite{E03,AH05,CW10}, the maximum $p$ for providing $\rho_1$-$\rho_2$ privacy is $p=\frac{\gamma-1}{m-1+\gamma}$, where $\gamma =\frac{\rho_2}{\rho_1}\times\frac{1-\rho_1}{1-\rho2}$ and $m=|SA|$. Additionally, RP has the parameters $\varepsilon$ and $\delta$ for $(\varepsilon,\delta)$-reconstruction-privacy.
We consider the settings of $p$, $\varepsilon$, $\delta$, and $|D|$, shown in
%$p$ for providing $\rho_1$-$\rho_2$ privacy is computed  by $p=\frac{\gamma-1}{m-1+\gamma}$, where $\gamma =\frac{\rho_2}{\rho_1}\times\frac{1-\rho_1}{1-\rho2}$.
Table \ref{table:parameters}. The default settings are in boldface. %$\delta \leq 0.5$ because the maximum probability for single-tailed distribution is 50\%.
%We consider EDU and OCC data sets separately.

%All algorithms were implemented in C++ and all experiments were conducted on an Intel Xeon(R) E5630 CPU 2.53GHZ PC with 12GB of RAM.

%\begin{figure}[h]
%\subfigure[Relative error vs $p$ and $|D|$]{
%\begin{minipage}[t]{0.45\linewidth}
%\centering
%\includegraphics[width=4cm,height=3cm]{figure/pComEDUError.eps}
%%% old a and b values for F'
%\end{minipage}}
%\hfill \subfigure[$\rho_1$-$\rho_2$ privacy specified by $p$]{
%\begin{minipage}[t]{0.45\linewidth}
%\centering
%\includegraphics[width=4cm,height=3cm]{figure/pEDU.eps}
%\end{minipage}}
%\caption{EDU:  Relative error, privacy, and $p$}
%%coefficient $\theta$ (x-axis)} \label{lvsa}
%\label{figure:pEDUError}
%\end{figure}

\subsubsection{Findings on EDU Data Sets}\label{EDU}
%Figure~\ref{figure:pEDUError} shows the relative error vs various $p$ and $|D|$, and the $\rho_1$-$\rho_2$ privacy specified by various $p$.
Figure~\ref{fig:perEDU} shows the percentage of violating micro
groups in $D^*$ vs $p$, $\varepsilon$, $\delta$, and $|D|$. Here are several observations. Firstly,
there is a nontrivial percentage of micro groups that violate $(\varepsilon,\delta)$-reconstruction-privacy. A larger retention probability $p$ leads to more violating micro groups. The violation diminishes when $p$ becomes very small (i.e., less than 20\%), but in this case aggregate reconstruction is affected significantly because $D^*$ is too noisy, as shown by the larger relative error. A larger $\varepsilon$ or $\delta$ leads to more violating micro groups due to a more restrictive privacy constraint. A larger data cardinality $|D|$ leads to more violating micro groups. This is because a larger $|D|$ leads to a larger $|g|$, i.e., more independent trials when generating $g^*$, thus, a more accurate reconstruction. In fact, $|g|\leq  =\frac{-2\ln \delta}{w\theta^2}$ is more likely to be violated as $|g|$ increases.

For each experiment in Figure~\ref{fig:perEDU}, Figure~\ref{fig:errEDU} shows the relative error of UP and RP. Note that, in Figure~\ref{fig:errEDU} (b)(c), UP remains constant because UP does not depend on $\varepsilon$ and $\delta$. The most significant finding is that, across all of $p, \varepsilon, \delta$, and $|D|$,
the error of RP is only slightly more than the error of UP. This point can also be seen by cross-examining Figure~\ref{fig:perEDU} and Figure~\ref{fig:errEDU}: the increase of error for RP is much slower than the increase in the percentage of violating micro groups. The reason is that the error boosting of RP through reducing the number of independent trails has less effect on queries that involve a large set of records.
This finding supports our claim that the proposed method does not compromise the utility of aggregate information. For $p$ and $|D|$, the trend in Figure~\ref{fig:perEDU} and Figure~\ref{fig:errEDU} is opposite: as $p$ or $|D|$ increases, the percentage of violating micro groups increases, but the error of estimated query answers decreases. This makes sense because violating micro groups are caused by high accuracy of estimated query answers.

\subsubsection{Findings on OCC Data Sets}\label{OCC}
\begin{figure*}[t!]
  \begin{center}
    \mbox{
      \subfigure[vs. $p$]{
        \includegraphics[width=4cm,height=3cm]{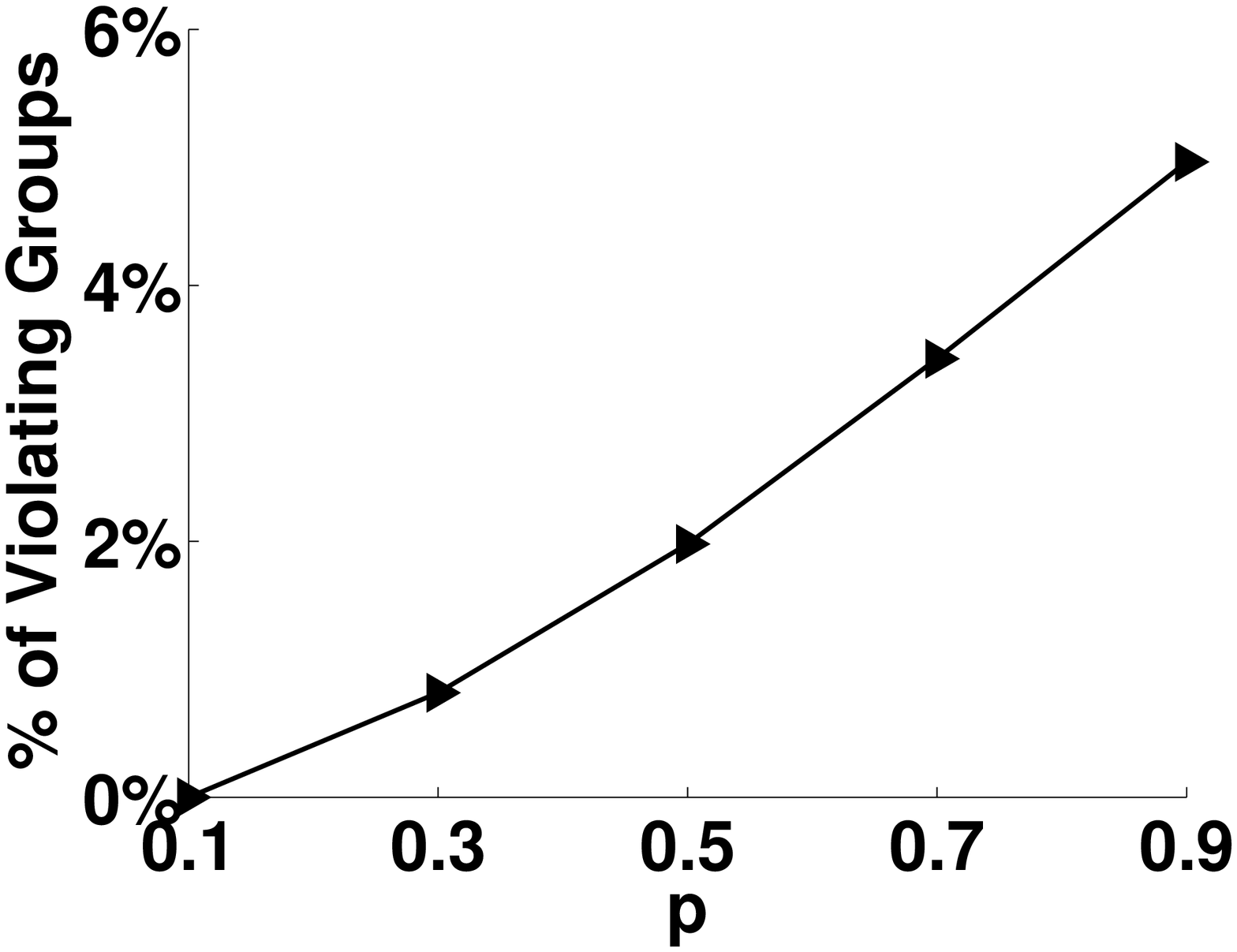}
      }
      \subfigure[vs. $\varepsilon$]{
        \includegraphics[width=4cm,height=3cm]{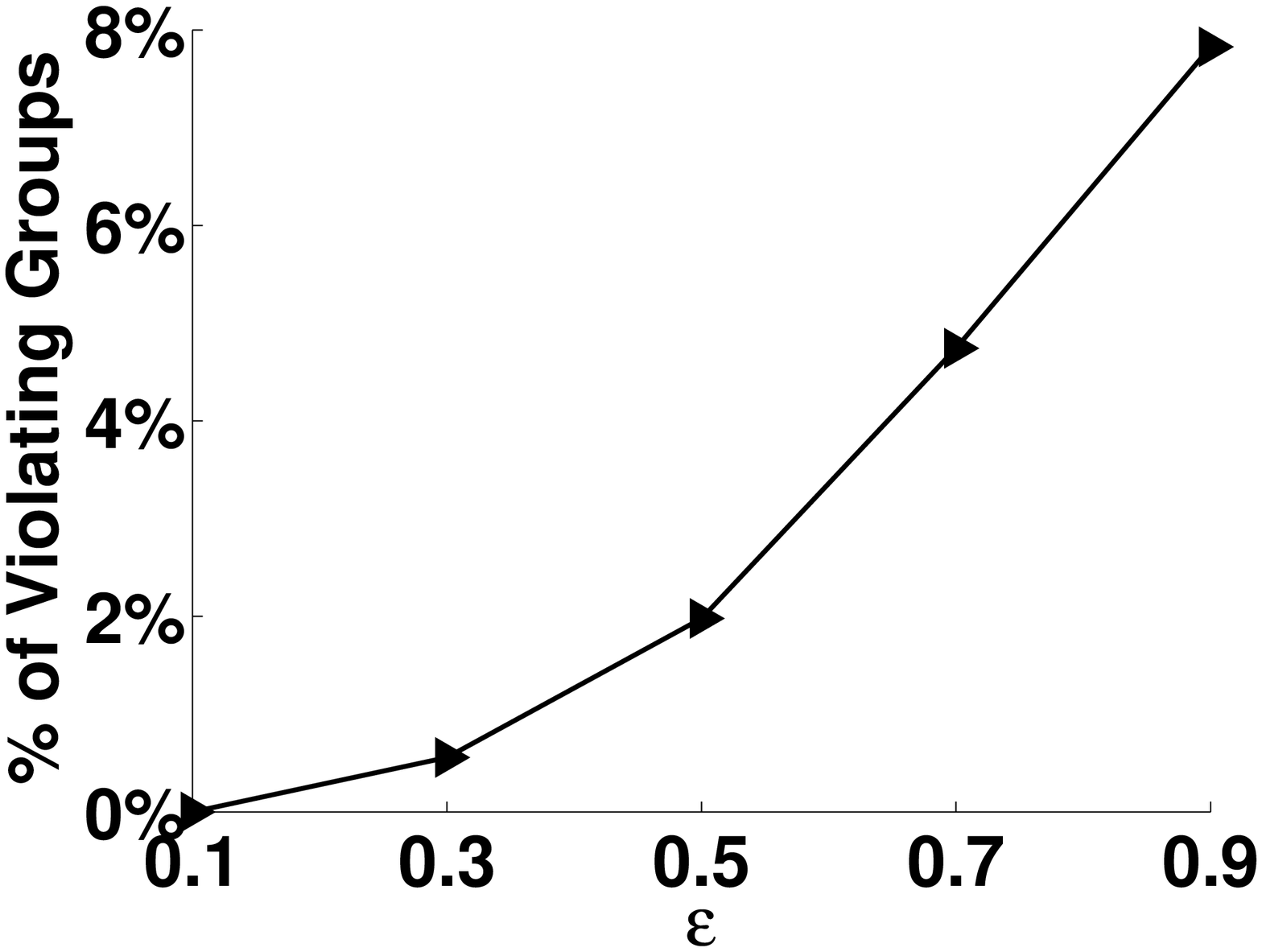}
        }
        \subfigure[vs. $\delta$]{
        \includegraphics[width=4cm,height=3cm]{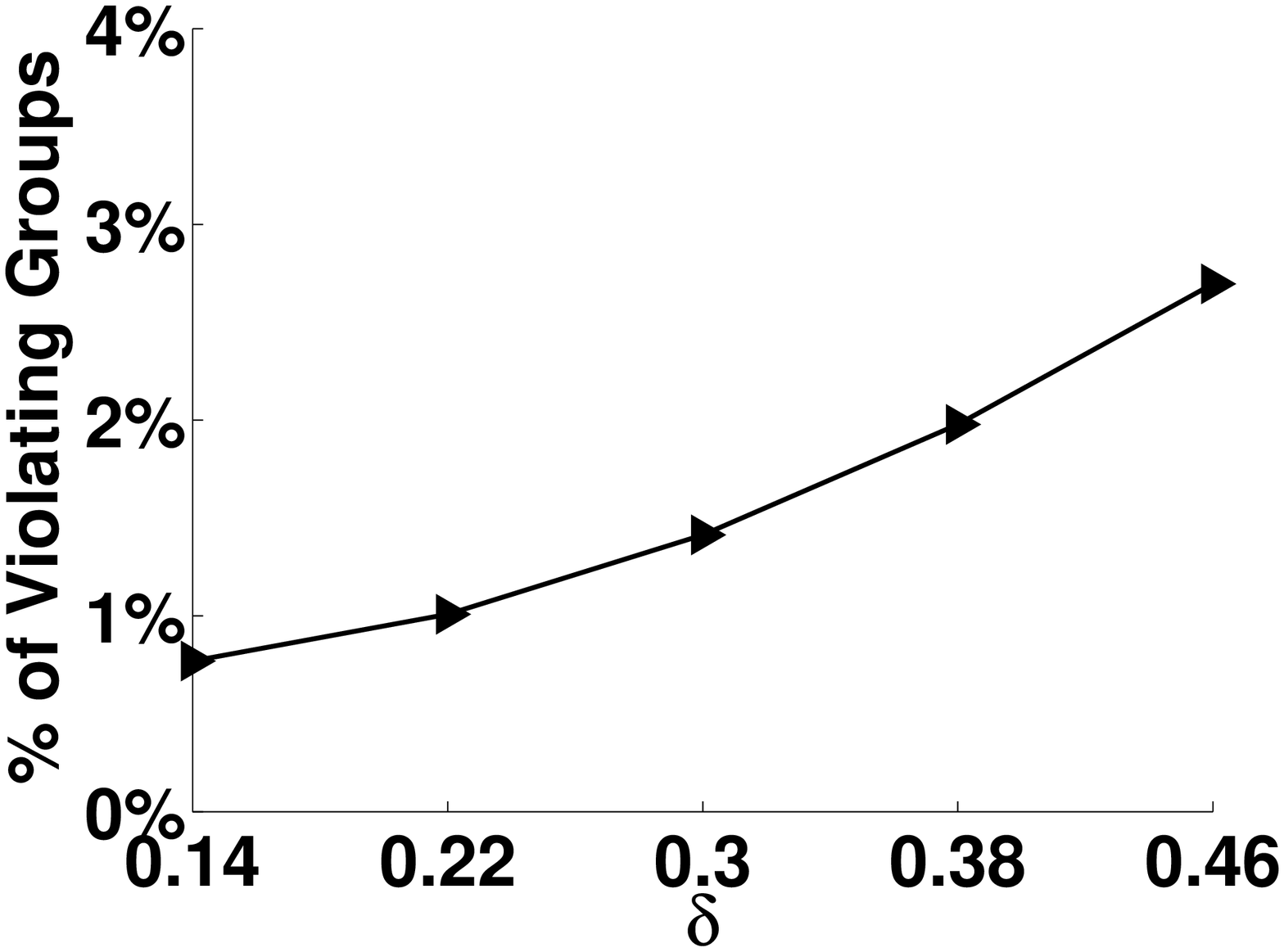}
      }
      }
      \subfigure[vs. $|D|$]{
        \includegraphics[width=4cm,height=3cm]{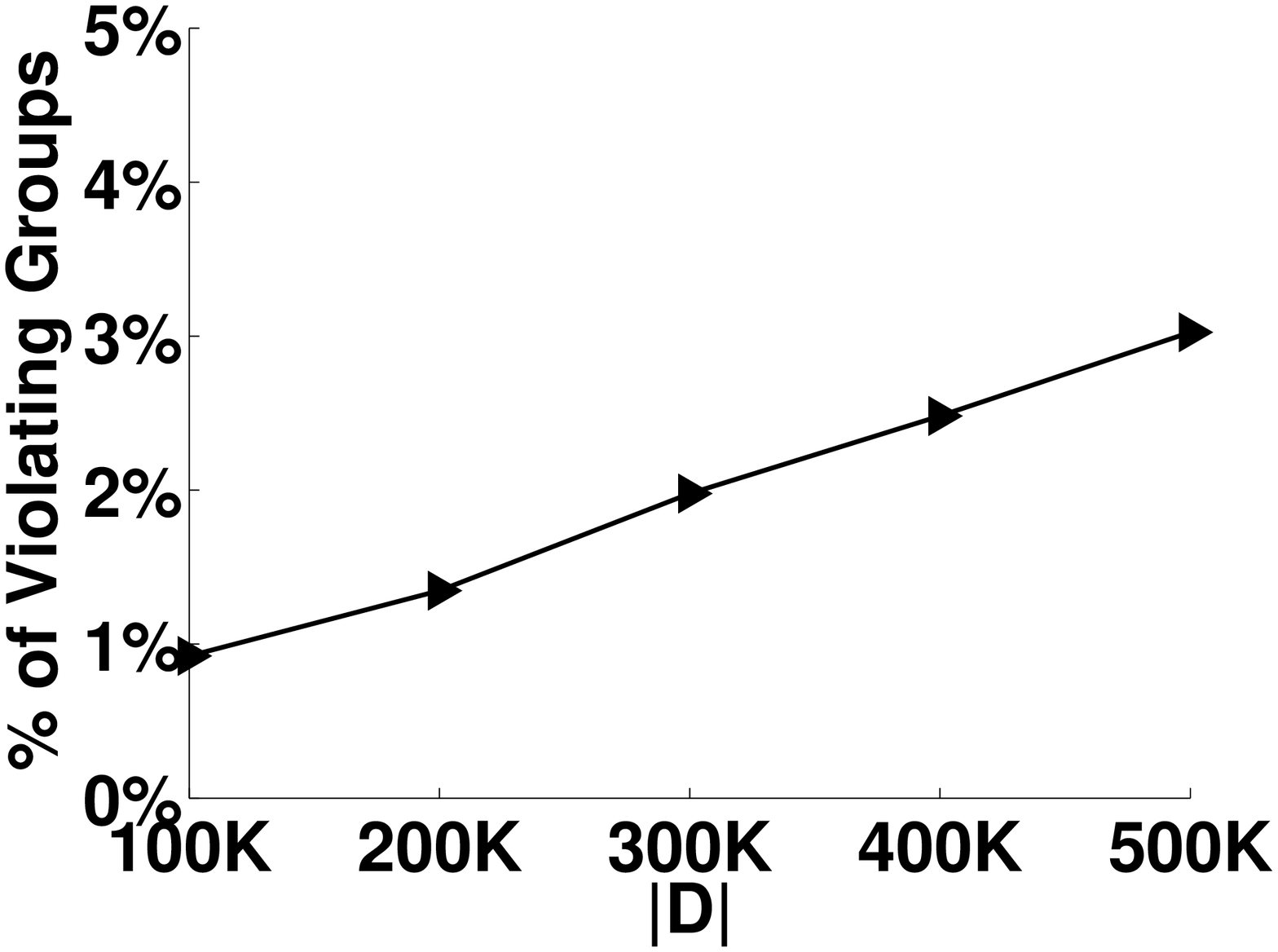}
      }
    \caption{OCC: \% of Violating Micro Groups in $D^*$}
    \label{fig:perOCC}
    \vspace{1cm}
    \mbox{
      \subfigure[vs. $p$]{
        \includegraphics[width=4cm,height=3cm]{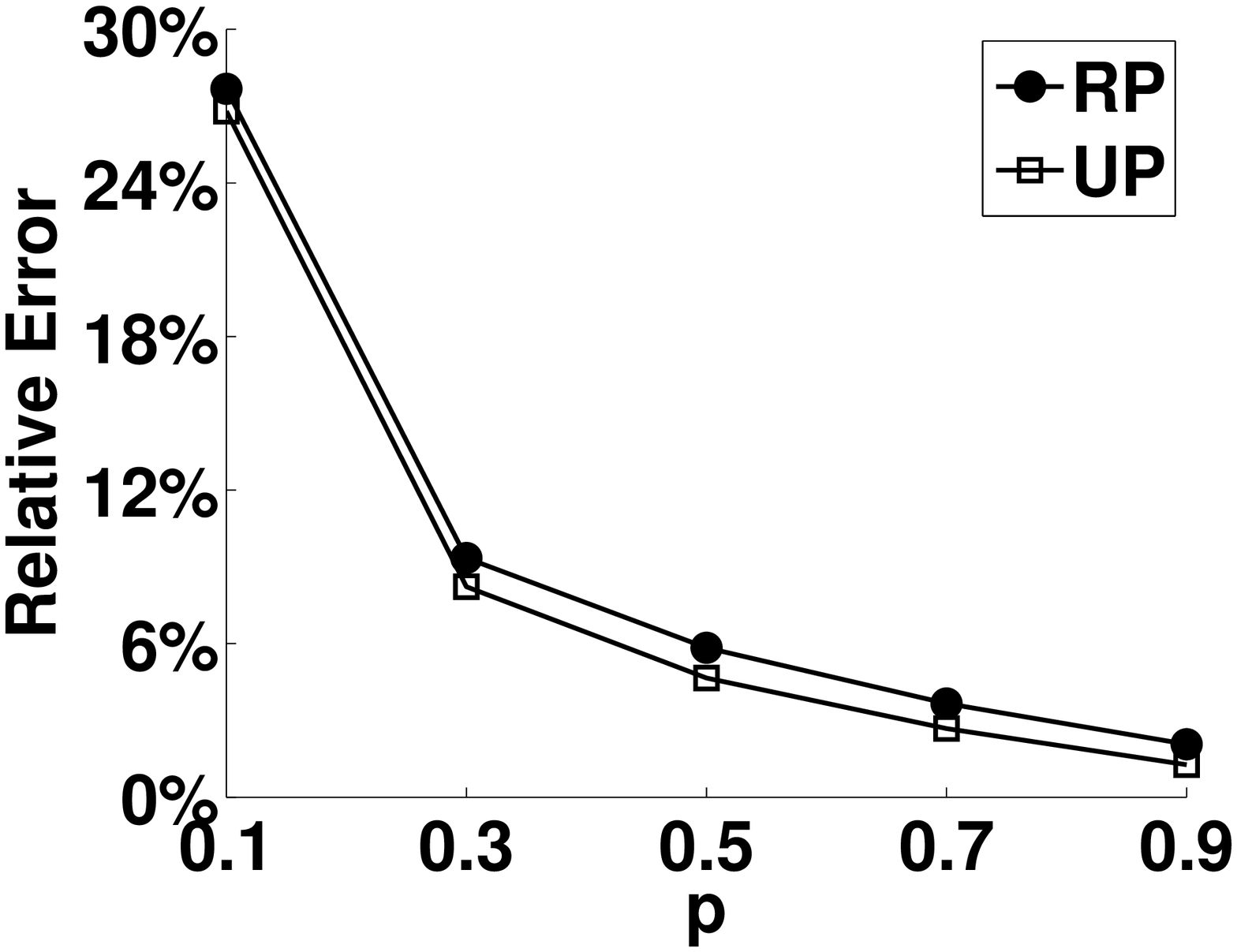}
      }
      \subfigure[vs. $\varepsilon$]{
        \includegraphics[width=4cm,height=3cm]{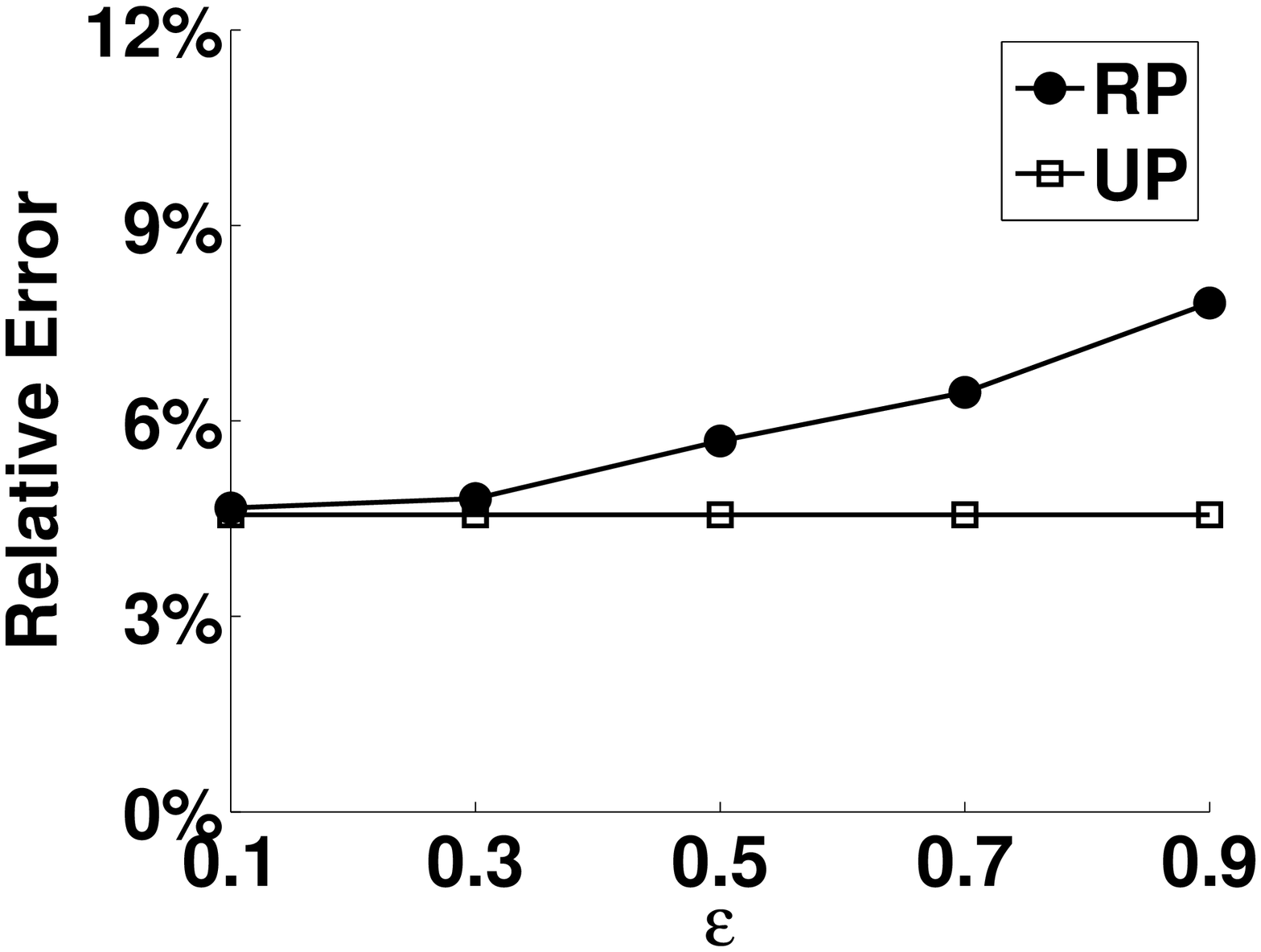}
        }
        \subfigure[vs. $\delta$]{
        \includegraphics[width=4cm,height=3cm]{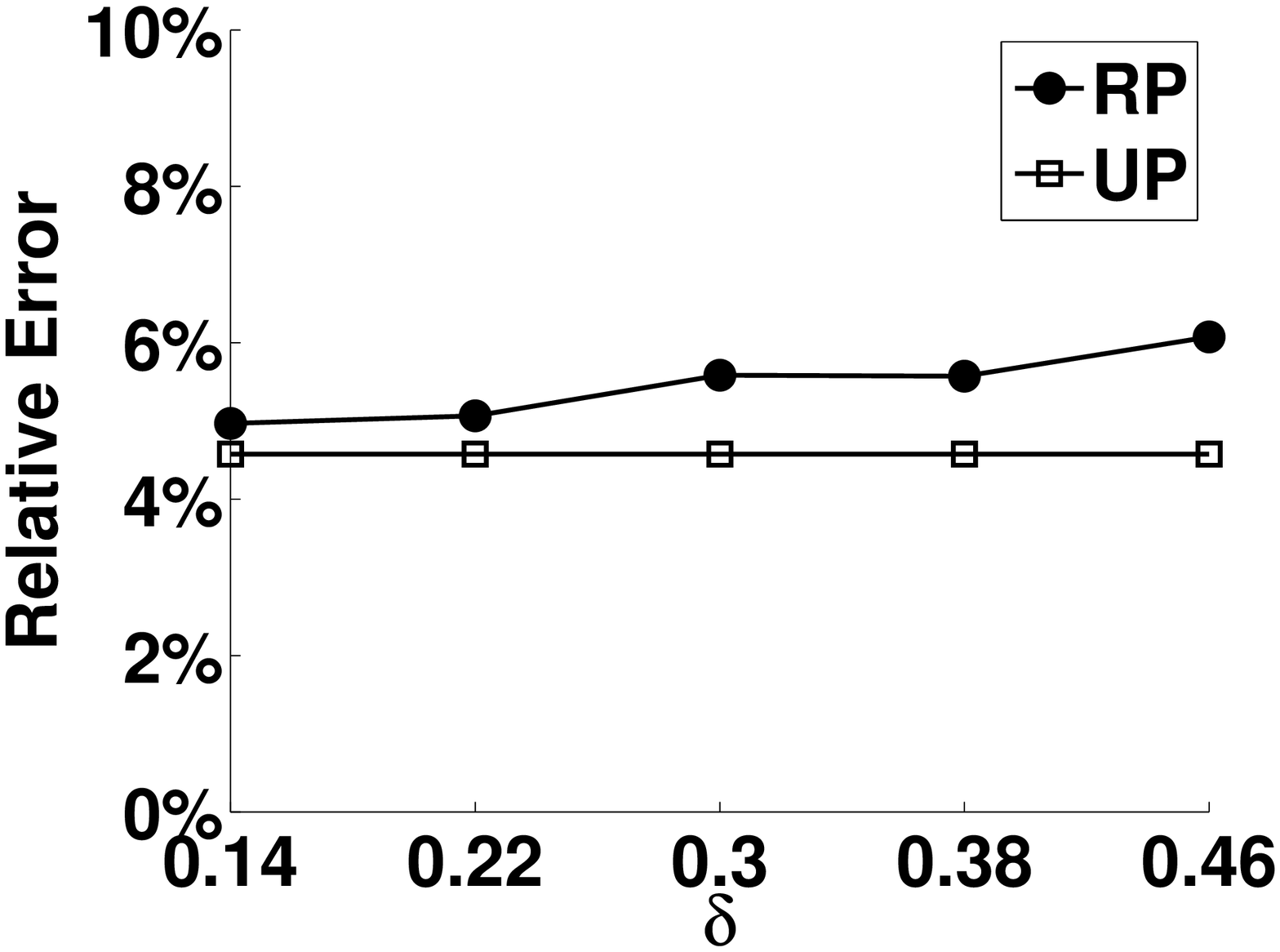}
      }
      }
      \subfigure[vs. $|D|$]{
        \includegraphics[width=4cm,height=3cm]{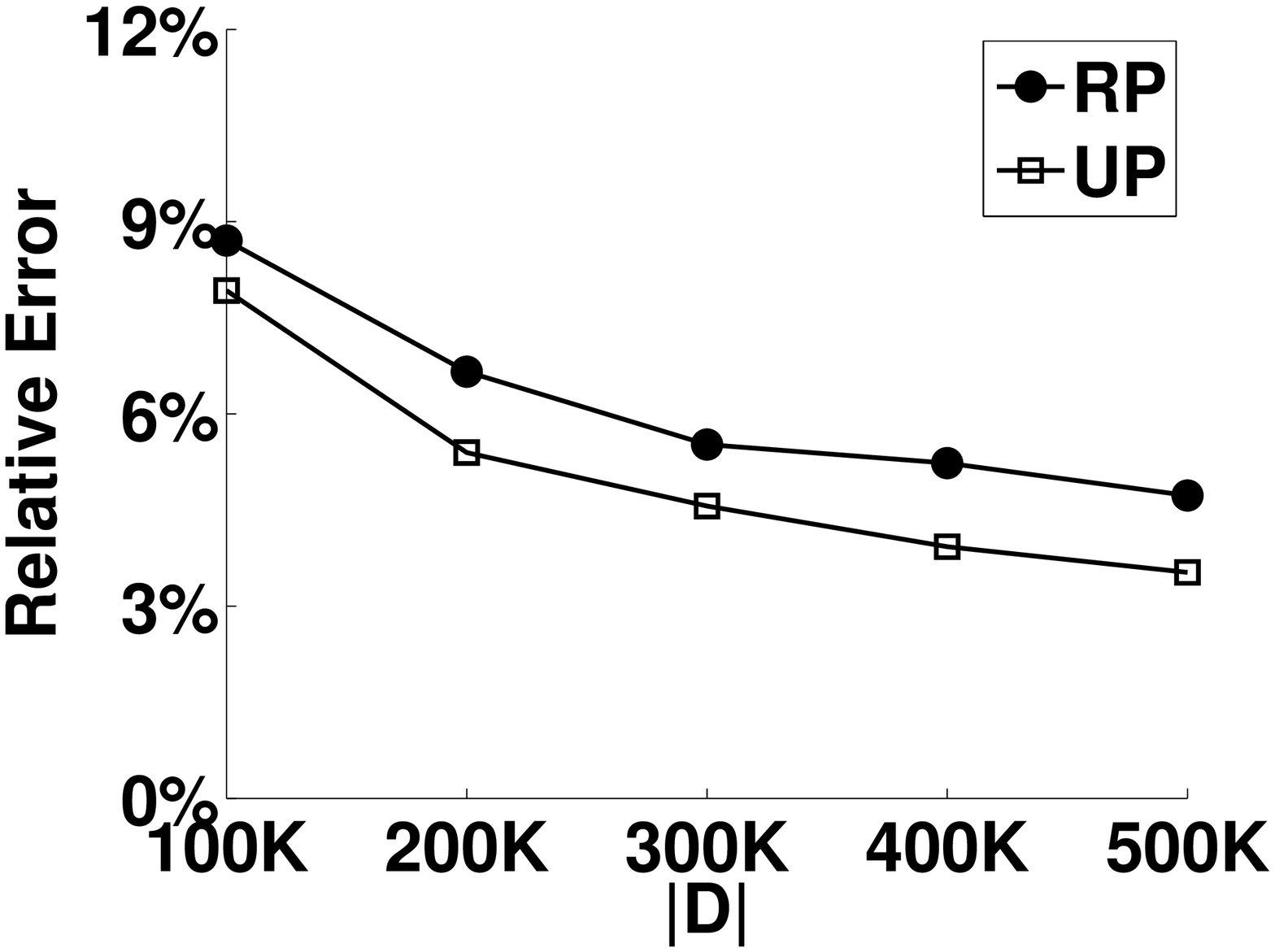}
      }
    \caption{OCC: Comparison of Relative Error for Count Queries}
    \label{fig:errOCC}
  \end{center}
\end{figure*}

We performed a similar study on the more balanced OCC data sets.
%%Figure~\ref{figure:pOCCError} shows the relative error vs various $p$ and $|D|$, and $\rho_1$-$\rho_2$ privacy specified by $p$.
Figure~\ref{fig:perOCC} shows the percentage of violating micro groups and Figure~\ref{fig:errOCC} shows the relative error, respectively. As we can see, the findings are quite similar to those of EDU data sets.
%
%\begin{figure}[h]
%\subfigure[Relative error vs $p$ and $|D|$]{
%\begin{minipage}[t]{0.45\linewidth}
%\centering
%\includegraphics[width=4cm,height=3cm]{figure/pComOCCError.eps}
%%% old a and b values for F'
%\end{minipage}}
%\hfill \subfigure[$\rho_1$-$\rho_2$ privacy specified by $p$]{
%\begin{minipage}[t]{0.45\linewidth}
%\centering
%\includegraphics[width=4cm,height=3cm]{figure/pOCC.eps}
%\end{minipage}}
%\caption{OCC: Relative error, privacy, and $p$}
%%coefficient $\theta$ (x-axis)} \label{lvsa}
%\label{figure:pOCCError}
%\end{figure}

\subsection{Findings on Output Perturbation}
Although Definition \ref{privacy} is based on reconstruction from a randomized data $D^*$, the notion of $(\varepsilon, \delta)$-reconstruction-privacy is applicable to any reconstruction. In this experiment, we consider reconstruction from noisy query answers in the output perturbation scenario.
We assume that differential privacy \cite{Dwork06} is in place. The $\lambda$-differential privacy mechanism adds random noises $\xi$ to the query answer $o$ and publishes the noisy answer $o'=o+\xi$, where $\xi$ follows the Laplace distribution $Lap(b)=\frac{1}{2b}exp(-\frac{|\xi|}{b})$, $b=1/\lambda$. $\lambda$ determines the noise level. We show that even if differential privacy is satisfied, there is a concern about violation of $(\varepsilon, \delta)$-reconstruction-privacy. We use EDU-500K and OCC-500K.

For each data set, we pick 7 micro groups $g$
that have the largest maximum frequency $f$ of any $SA$ value, among those with $|g|>70$ for EDU-500K and $|g|>100$ for OCC-500K. For each of these groups, $g$, let $f$ and $F'$ be the true and estimated frequencies of the most frequent $SA$ value $x$ in $g$. $F'$ is computed by the noisy answers to two queries $Q_1$ and $Q_2$ constructed similar to those in Example \ref{example2}. Let $o_i$ be the true answer and let $o'_i$ be the noisy answer for $Q_i$, $i=1,2$. $f=o_2/o_1$ and $F'=o'_2/o'_1$. By treating $\Pr[\frac{F'-f}{f}> \varepsilon]$  and $\Pr[\frac{F'-f}{f} < -\varepsilon]$ as the upper bounds of these probabilities themselves, $(\varepsilon, \delta)$-reconstruction-privacy is violated if $\Pr[\frac{F'-f}{f} > \varepsilon] < \delta$ or $\Pr[\frac{F'-f}{f} < -\varepsilon] < \delta$. To compute these probabilities, we generated the noisy answers $o'_1$ and $o'_2$ 100 times and considered the fraction of the cases for $\frac{F'-f}{f}> \varepsilon$ and $\frac{F'-f}{f} < -\varepsilon$. The numbers for these cases are  in Tables~\ref{table:diffEDU} and \ref{table:diffOCC}.

%Table~\ref{table:diffEDU} shows the result of error-cases in data set EDU-500K and table~\ref{table:diffOCC} shows that in data set OCC-500K. Noting that all $f$ are above $0.7$ in table~\ref{table:diffEDU} and above $0.4$ in table~\ref{table:diffOCC}. This is because OCC is a data set with relatively more balanced distribution and it is not likely for OCC data sets to have a dominant SA in a micro group.
%
%
%In this experiment smaller number of error cases can be translated as more chance to violate our $(\varepsilon, \delta)$-reconstruction-privacy. Takeing group 2 in table~\ref{table:diffOCC} as an example. When $\lambda=0.1$ and $\varepsilon=0.3$ out of 100 times there are only 3 times that the relative error is larger than $\varepsilon$, all the rest are incurring a absolute relative error smaller than $\varepsilon$. This promises a much accurate reconstruction and will help the adversary to issue the "micro reconstruction attack" successfully.

Take Group 7 in Table~\ref{table:diffEDU} (in boldface) for EDU-500K as an example. For $\lambda=0.1$ and $\varepsilon=0.3$, there are 8 cases for $>\varepsilon$ and 8 cases for $<-\varepsilon$. Intuitively, this says that, out of the 100 noisy answers $(o'_1,o'_2)$ examined, 8 cases have an error greater than 30\% and 8 cases have an error less than $-30\%$. In other words, the estimate $F'$ falls within the $\pm 30\%$ interval with the confidence level of 84\%. The privacy concern comes from the fact that the frequency of $x$ in $g$ is more than 70\% (shown in the column ``$f$ in $g$"), which is significantly higher than the 2.5\% in the whole data set $D$ (shown in the column ``$f$ in $D$"). Thus, even if the $\pm 30\%$ interval is large, $F'$ discloses a much higher probability of having $x$ for the individuals in $g$ than for the individuals in $D$. Similar disclosures are observed on the more balanced OCC-500K. For $\lambda=0.1$ and $\varepsilon=0.2$, Group 2 in Table~\ref{table:diffOCC} (in boldface) shows a $\pm 20\%$ error interval with the confidence level of 84\%. Although the frequency $f$ of $x$ in this group is only 47\%, it is significantly higher than the frequency of 2.4\% in the whole data set. Therefore, $F'$ discloses quite a bit about the $SA$ value of the individuals in this group.

At $\lambda=0.05$, a larger error for $F'$ has been observed due to the increased noise level. However, since $\lambda$ is a constant for a given $\lambda$-differential privacy mechanism, the error for $F'$ can be reduced by a sufficiently large group size $|g|$ and frequency $f$ in $g$. To provide $(\varepsilon,\delta)$-reconstruction-privacy, the $\lambda$-differential privacy mechanism has to employ a very small $\lambda$. This solution shares the same drawback with the solution of using a small retention probability $p$, i.e., choosing the global noise parameters, i.e., $\lambda$ and $p$, according to the \emph{worst case} of any micro group in the data set. As discussed in Section 6.2.1, this type of solutions destroys both micro reconstruction and aggregate reconstruction, making the data useless for \emph{all} queries.

%To summarize, this study shows that the reconstruction attack could occur even if differential privacy is ensured.

%Several points can be observed from this experiment. Firstly, more error cases are incurred as $\lambda$ decreases. The intuitive here is that small $\lambda$ leads to more noise is added to the answer of queries, which makes the reconstruction not "too accurate". Take group 7 in table~\ref{table:diffEDU} as an example. when $\lambda=0.1$, $\varepsilon=0.2$ the number of $>\varepsilon$ is $18$ and the number of $<-\varepsilon$ cases is $13$.
%
%
%As $\lambda$ drops to $0.05$, these two numbers increase to $32$ and $31$ under the same $\varepsilon$. Secondly, smaller $\varepsilon$ will result in more error cases. The reason is obvious that it is easier to have more cases with error larger than a $\varepsilon$ when $\varepsilon$ is relatively smaller. Still using group 7 as the example, with the same $\lambda=0.1$, the number of pos-error cases and the number of neg-error cases decreases from 18 and 13 when $\varepsilon=0.2$ to 8 and 8 when $\varepsilon=0.3$. Finally, table~\ref{table:diffOCC} shows the similar trend.

\begin{table*}[t!]
\centering
\caption{EDU-500K: the Number of Cases for  $\frac{F'-f}{f}> \varepsilon$ and $\frac{F'-f}{f} < -\varepsilon$}
\begin{tabular}{|c|c|c|c|c|c|c|c|c|c|c|c|} \hline
\multirow{3}{*}{Micro Group g} & \multirow{3}{*}{$|g|$} & \multirow{3}{*}{$f$ in $g$} & \multirow{3}{*}{$f$ in $D$}& \multicolumn{4}{|c|}{$\lambda=0.1$} & \multicolumn{4}{|c|}{$\lambda=0.05$} \\ \cline{5-12}
& & & & \multicolumn{2}{|c|}{$\varepsilon=0.2$} & \multicolumn{2}{|c|}{$\varepsilon=0.3$} & \multicolumn{2}{|c|}{$\varepsilon=0.2$} & \multicolumn{2}{|c|}{$\varepsilon=0.3$} \\ \cline{5-12}
& & & & {$> \varepsilon$} & {$<-\varepsilon$} & {$>\varepsilon$} & {$<-\varepsilon$} & {$> \varepsilon$} & {$< -\varepsilon$} & {$> \varepsilon$} & {$< -\varepsilon$}\\ \hline
1 & 89 & 0.87 & 0.025 & 18 & 13 & 14 & 7 & 34 & 29 & 24 & 22\\ \hline
2 & 74 & 0.77 & 0.025 & 25 & 23 & 14 & 7 & 32 & 32 & 28 & 23\\ \hline
3 & 138 & 0.76 & 0.172 & 11 & 9 & 8 & 3 & 18 & 27 & 20 & 15\\ \hline
4 & 104 & 0.76 & 0.172 & 21 & 12 & 9 & 6 & 35 & 28 & 22 & 22\\ \hline
5 & 104 & 0.75 & 0.172 & 23 & 14 & 11 & 6 & 35 & 25 & 21 & 28\\ \hline
6 & 77 & 0.74 & 0.025 & 26 & 11 & 18 & 11 & 26 & 39 & 27 & 26\\ \hline
\textbf{7} & \textbf{102} & \textbf{0.72} & \textbf{0.025} & \textbf{18} & \textbf{13} & \textbf{8} & \textbf{8} & \textbf{32} & \textbf{31} & \textbf{29} & \textbf{21}\\ \hline
%8 & 108 & 0.712963 & 14 & 16 & 11 & 5 & 36 & 28 & 24 & 25\\ \hline
%9 & 127 & 0.708661 & 22 & 8 & 7 & 7 & 23 & 25 & 21 & 12\\ \hline
%10 & 151 & 0.701987 & 17 & 6 & 7 & 5 & 28 & 24 & 17 & 9\\ \hline
\end{tabular}
\label{table:diffEDU}
%\end{table*}
%\begin{table*}[t]
\centering
\caption{OCC-500K: the Number of Cases for  $\frac{F'-f}{f}> \varepsilon$ and $\frac{F'-f}{f} < -\varepsilon$}
\begin{tabular}{|c|c|c|c|c|c|c|c|c|c|c|c|}\hline
\multirow{3}{*}{Micro Group g} & \multirow{3}{*}{$|g|$} & \multirow{3}{*}{$f$ in $g$} & \multirow{3}{*}{$f$ in $D$} & \multicolumn{4}{|c|}{$\lambda=0.1$} & \multicolumn{4}{|c|}{$\lambda=0.05$} \\ \cline{5-12}
& & & & \multicolumn{2}{|c|}{$\varepsilon=0.2$} & \multicolumn{2}{|c|}{$\varepsilon=0.3$} & \multicolumn{2}{|c|}{$\varepsilon=0.2$} & \multicolumn{2}{|c|}{$\varepsilon=0.3$} \\ \cline{5-12}
& & & & {$> \varepsilon$} & {$<-\varepsilon$} & {$>\varepsilon$} & {$<-\varepsilon$} & {$> \varepsilon$} & {$< -\varepsilon$} & {$> \varepsilon$} & {$< -\varepsilon$}\\ \hline
%& & & {number of pos-error cases} & {number of neg-error cases} & {number of pos-error cases} & {number of neg-error cases} & {number of pos-error cases} & {number of neg-error cases} & {number of pos-error cases} & {number of neg-error cases}\\ \hline
1 & 142 & 0.48 & 0.038 & 13 & 18 & 11 & 4 & 29 & 29 & 27 & 20\\ \hline
\textbf{2} & \textbf{213} & \textbf{0.47} & \textbf{0.024} & \textbf{8} & \textbf{8} & \textbf{3} & \textbf{0} & \textbf{23} & \textbf{19} & \textbf{15} & \textbf{11}\\ \hline
3 & 111 & 0.47 & 0.026 & 26 & 18 & 16 & 13 & 40 & 30 & 27 & 29\\ \hline
4 & 113 & 0.45 & 0.024 & 28 & 20 & 17 & 8 & 38 & 30 & 23 & 25\\ \hline
5 & 153 & 0.45 & 0.026 & 18 & 15 & 6 & 6 & 20 & 40 & 26 & 21\\ \hline
6 & 237 & 0.45 & 0.024 & 8 & 9 & 3 & 3 & 17 & 21 & 13 & 12\\ \hline
7 & 143 & 0.44 & 0.038 & 12 & 17 & 12 & 6 & 38 & 34 & 27 & 20\\ \hline
%8 & 175 & 0.44 & 13 & 13 & 7 & 5 & 21 & 27 & 22 & 24\\ \hline
%9 & 200 & 0.44 & 12 & 7 & 6 & 6 & 25 & 15 & 19 & 20\\ \hline
%10 & 169 & 0.701987 & 10 & 19 & 10 & 6 & 34 & 23 & 19 & 17\\ \hline
\end{tabular}
\label{table:diffOCC}
\end{table*}

%\subsection{Summary}\label{eConclusion}

\section{Conclusion}
Reconstruction of data distribution is traditionally regarded as utility. In this work, we showed that reconstruction could lead to privacy breaches even if major privacy definitions are satisfied.
We formalized a privacy definition to address this risk and presented an enforcement solution. A novelty of this work lies at the distinction between reconstruction that has privacy risk and reconstruction that does not. We leveraged this distinction to meet the dual requirement of privacy and utility. Another novelty is the independence on the particular form of the bounds on tail probabilities. Our privacy definition is a constraint on the upper bounds of tail probabilities and the Chernoff bound in particular. This formulation allows us to leverage the upper bound literature to develop a concrete solution to the problem identified. However, our approach can be instantiated to other upper bounds and modified to constrain the lower bounds of tail probabilities, thanks to the general form of the bound conversion theorem (Theorem \ref{UL-bound}).

\bibliography{vldb}

\end{document}